\documentclass[a4paper,11pt]{article}
\usepackage[utf8]{inputenc}
\usepackage[margin=.9in, top=1in, textwidth=5.5in]{geometry}

\usepackage{amsmath,amssymb,amsthm}
\usepackage{authblk}
\usepackage{amsfonts}
\usepackage{bm}
\usepackage{dsfont}

\usepackage{graphicx}
\usepackage{hyperref}
\hypersetup{colorlinks=true,linkcolor=blue}
\usepackage{amssymb}
\usepackage{wrapfig}
\usepackage{caption}
\usepackage[usenames,dvipsnames]{xcolor}
\usepackage{subcaption}
\usepackage[linesnumbered, ruled,vlined]{algorithm2e}
\usepackage[square,numbers]{natbib}
\SetKwRepeat{Do}{do}{while}

\newtheorem{theorem}{Theorem}
\newtheorem{claim}{Claim}

\newtheorem{definition}{Definition}

\newtheorem{conjecture}{Conjecture}
\newtheorem{lemma}{Lemma}

\newcommand{\rpskg}{{\tt RPSKG}~}
\newcommand{\Prob}[1]{\mathbb{P}(#1)}

\def\reals{{\mathbb R}}
\def\naturals{{\mathbb N}}
\def\rationals{{\mathbb Q}}

\newcommand{\E}{\mathbb{E}}
\newcommand{\pr}{\mathbb{P}}

\newcommand{\tv}{\mathrm{TV}}

\newcommand{\Var}{\mathrm{Var}}

\newcommand{\poly}{\mathrm{poly}}

\newcommand{\ind}{\mathds{1}}

\def\calA{\mathcal{A}}
\def\calB{\mathcal{B}}
\def\calD{\mathcal{D}}
\def\calF{\mathcal{F}}
\def\calP{\mathcal{P}}

\begin{document}

\title{Degree Distribution Identifiability of Stochastic Kronecker Graphs}

\author[1]{
Daniel G. Alabi
\footnote{Supported by the Simons Foundation (965342, D.A.) as part of the 
Junior Fellowship from the Simons Society of Fellows
and a Fellowship from Meta AI during his
graduate studies.
}
}
\affil[1] {
Columbia University
}

\author[2]{
Dimitris Kalimeris
\footnote{Work done while a Ph.D. student at Harvard University.}
}
\affil[2] {
Meta
}

\date{}

\maketitle

\begin{abstract}

Large-scale analysis of the distributions of the network graphs observed in
naturally-occurring phenomena has revealed that the degrees of such graphs follow a power-law or lognormal distribution.
Some network or graph generation algorithms rely on the
Kronecker multiplication primitive. Seshadhri, Pinar, and Kolda (J. ACM, 2013) proved that
stochastic Kronecker graph (SKG) models \textit{cannot} generate graphs with degree distribution that follows
a power-law or lognormal distribution. As a result, variants of the SKG model have been proposed
to generate graphs which \textit{approximately} 
follow degree distributions, without any significant oscillations. However, all
existing solutions either require significant additional parameterization or have no provable guarantees on the degree
distribution.

\begin{itemize}
\item In this work, we present
statistical and computational identifiability notions
which imply the
\textit{separation} of SKG models.
Specifically, we prove that SKG models in different
identifiability classes can be separated by the existence of
isolated vertices and connected components in their
corresponding generated graphs.
This could explain the large (i.e., $>50\%$)
fraction of isolated vertices
in some popular graph generation benchmarks.

\item We present and analyze an efficient algorithm that can
get rid of
oscillations in the degree distribution by mixing seeds
of relative prime dimensions. For an initial $2\times 1$
and $2\times 2$ seed,
a crucial subroutine of
this algorithm solves a
degree-2 and degree-4 optimization problem in the variables of the
initial seed, respectively.
We generalize this approach to solving optimization problems
for $m\times n$ seeds, for any $m, n\in\naturals$.

\item The use of $3\times 3$ seeds alone cannot
get rid of significant oscillations. We prove that such
seeds result in degree distribution that is bounded above
by an exponential tail and thus cannot result in a
power-law or lognormal.

\end{itemize}

The definitions of identifiability of SKG models 
could lead to a better understanding and classification
of graph generation algorithms.

\end{abstract}

\clearpage

\tableofcontents

\clearpage

\section{Introduction}

Graph models can be used to
represent biological networks,
social networks, communication networks, relationships
between products and advertisers, and so on~\citep{ClausetSN09, LeskovecCKF05, MorenoNK18, doi:10.1073/pnas.0701175104, KumarNT06}.
Representing these networks as graphs enables the mining of
properties from these graphs~\citep{ChakrabartiF06}. However,
companies that store large graph data, such
as Netflix and Meta/Facebook, cannot share most of their
graph data because of copyright, legal, and
privacy issues~\citep{NarayananS08, ChettyF19, Chetty22}.
As a result, a lot of research has been done to enable generation of graph models that realistically model
``real-world'' graphs, which should have an
approximately lognormal or power law degree
distribution~\citep{KoPiPlSeTa14}.
The Stochastic Kronecker Graph (SKG) model
has been proposed as a graph generation model and is
currently employed by the Graph500 supercomputer
benchmark, particularly because the SKG graph
generation process is easily parallelized~\citep{Graph500}.
However,
the original SKG model \textit{cannot} generate a power law or
even a log-normal distribution but instead 
is most accurately characterized as fluctuating
between a lognormal distribution and an exponential tail~\citep{KimL12}. Alternatives like the Noisy Stochastic
Kronecker Graph model are also
unsatisfactory partly due to their
unwieldy additional parameterization~\citep{Chakrabarti:2004}.

We are thus motivated by the following questions posed by~\citep{Seshadhri:2013}:
\begin{quote}
\textit{
What does the [degree] distribution oscillate between? Is the distribution
bounded below by a power law? Can we approximate the distribution with a simple closed
form function? None of these questions have satisfactory answers.}
\end{quote}

Instead of attempting to unconditionally show that certain
SKG models have certain degree distributions in certain
parameter regimes, we begin a systematic classification of
SKG models. We view the model generating process through the
lens of statistical and computational 
identifiablility~\citep{keener2010theoretical, arora2006computational, BCL11}.

Definition~\ref{def:stat} presents a notion that can be
used to classify when an SKG model or algorithm follows
a certain distribution (e.g., lognormal or exponential).
Definition~\ref{def:comp} is the computational analogue.
We conjecture that certain SKG models cannot be SKG-identifiable
with certain parameterizations of the
lognormal distribution while others can be.
Then we discuss implications of these conjectures on the
following graph properties:
the existence of
isolated vertices, the presence of connected components, and
the degree distributions.
As noted in previous work~\citep{Seshadhri:2013},
the majority (i.e., $50-70\%$)
of vertices in popular benchmark graphs are isolated. This is a 
major concern for such benchmarks since the generated graph would have
a much smaller size than intended and the average degree would be
higher than intended.
Our work offers another view on the presence or
absence of isolated vertices across benchmark graphs. Such graph
generation algorithms belong in different identifiability classes
since by Theorem~\ref{thm:intro:isolated}, it is hard (in a precise
sense) to identify seeds that generate graphs that do not have
isolated vertices.

In addition, we present an algorithm that interpolates initiator
seeds of different dimensions to generate graphs with no
significant oscillations. We conjecture that SKG models require
more than one $2\times 2$
seed to remove significant oscillations in the
degree distribution. Evidence is presented to
support this claim.

A SKG model is defined by an algorithm $\calA$, an initiator set of
seeds (usually a $2\times 2$ or $3\times 3$ matrix), and the
target size of the graph. The Kronecker
product is then applied to the seed to define a probability
distribution over the resulting graph. See
Sections~\ref{sec:additionalback} and~\ref{sec:related}
for more details on using
the Kronecker multiplication primitive to generate a graph from
a seed.

Intuitively, an SKG model is \textit{identifiable} if there
exists at least one distribution that is close to the degree
distribution of the graph corresponding to an initiator seed.
The model is \textit{computationally identifiable} if, for any
distribution, there exists an algorithm that will
find such an initiator seed, that will correspond to the distribution, in a computationally-efficiently manner!
A new approach to the classification and formalization of
SKG models is necessary since it is not clear how to
get rid of the significant variability in the degree distributions
of the model. Furthermore, the identifiability criteria still
carries over the benefits of existing SKG model variants (such
as the Noisy SKG model and the Relative Prime Stochastic Kronecker
Graph model, both of which we will discuss in detail in later
sections). The separations of statistical and computational
identifiability is also necessary because, except for
ad-hoc and brute-force approaches, 
no \textit{feasible} procedure
is known to obtain or
\textit{infer} a seed that can be used to generate a graph
with a certain degree distribution. Also, while the focus
of existing SKG investigations has been on graphs with
degree distributions that follow a power-law
or lognormal measure, the definitions presented below aim to
generalize the study of SKG target degree distributions.

For the definitions below,
$\calP(\Omega, \calF)$ is the set of all probability measures
defined on $(\Omega, \calF)$.

\begin{definition}[SKG Identifiable]

Fix $\alpha\geq 0$.
Let $(\Omega, \calF)$ be a (Borel) measurable space and
$\calD$ be any probability measure or distribution defined on
$(\Omega, \calF)$.
For any $s\in\naturals$,
consider $\calA:\reals^{s\times s}\times\naturals\rightarrow\calP(\Omega, \calF)$, a
(randomized) algorithm that outputs a degree
distribution, given a set of seeds and target size of graph.

Then $\calA$ is 
\textbf{$\alpha$ SKG identifiable} with distribution $\calD$ if
there exists a set of seeds $S$\footnote{The set of seeds is
drawn from $\reals^{s\times s}$.}
such that for all $n\in\naturals$, 
$\tv(\calA(S, n), \calD) \leq \alpha$.

\label{def:stat}
\end{definition}

For Definition~\ref{def:stat}, we do not need to
require that $\calA$ output a graph if it outputs a degree distribution
since, up to isomorphisms, every degree distribution is realizable by
a graph of a certain size.
$\tv$ is the total variation distance (or statistical distance)
between two distributions.
\footnote{
For any measurable space $(\Omega, \calF)$ and probability
measures $P$ and $Q$ defined on $(\Omega, \calF)$,
the total variation is
$\tv(P, Q) = \sup_{B\in\calF}|P(B) - Q(B)|$.
}
Inspired by the computational analog of
statistical distance
(as introduced by Goldwasser and 
Micali~\citep{GOLDWASSER1984270, GM82}), we introduce a
computational analogue of identifiability of stochastic
Kronecker graphs.
We note that other statistical measures (e.g., max divergence or R\'enyi divergence) might
be better at capturing tail behavior of distributions which could be quite interesting. We leave
the exploration of such additional measures to future work.

\begin{definition}[SKG Computationally Identifiable]

Let $\alpha(n) \geq 0$ be a sequence and
let $(\Omega, \calF)$ be a (Borel) measurable space and
$\calD$ be any probability measure or distribution defined on
$(\Omega, \calF)$ that can be represented by $\poly(n)$-bit strings.
Let $\calB = \{\calB_n\}_{n\in\naturals}$ be a sequence of
(randomized) algorithms that each can be represented by
$\poly(n)$-bit strings.
For any $s\in\naturals$ and $n\in\naturals$,
consider $\calB_n:\naturals\times\calP(\Omega, \calF)\rightarrow\reals^{s\times s}\times \calP(\Omega, \calF)$, a
(randomized) algorithm that, given any
target size of graph and target distribution,
can output a seed that can generate a graph with degree distribution
that is close to the target.

Then $\calB$ is 
\textbf{$\alpha$ SKG computationally identifiable}
with distribution $\calD$ if
$\calB$ runs in time $\poly(n, 1/\alpha)$ and
$\calB(n, \calD)$ outputs a set of seeds $S$ and distribution
$\calD_n$ (representable by $\poly(n)$-bit strings) such that
for all $n\in\naturals$, 
$\tv(\calD_n, \calD) \leq \alpha(n)$.

\label{def:comp}
\end{definition}

Definition~\ref{def:stat} considers the existence of a seed that
can be used to generate a graph with
degree distribution that matches the
target. To meet the definition, the algorithm is not 
\textit{required} to find the seed. If such a seed exists, then
a brute-force search over the domain of seeds,
$\reals^{s\times s}$, can be used to obtain such a seed. 
However, in Definition~\ref{def:comp}, we essentially restrict
the set of algorithms to feasible algorithms that can find the
seed that satisfies the distribution closeness requirement.
That is, not only must the seed exist to generate the desired
degree distribution but the algorithm must be able to find such a
seed!

Definition~\ref{def:comp} bears resemblance to a one-way
function~\citep{Goldreich00, Goldreich04a}: it should be relatively easy to
perform the stochastic Kronecker multiplication to generate a graph
but obtaining the seed that would lead to \textit{any} degree
distribution should be computationally difficult.
This simple observation leads to one of our conjectures and the
resulting implication, which will be discussed in the next section.

To gain more intuition for statistical identifiability, consider 
Figure~\ref{fig:intro} that shows instantiations of the
Relative Prime Stochastic Kronecker Graph (\rpskg) algorithm
(Algorithm~\ref{alg:rpskg}) with
different parameters. Clearly, one instantiation results in
significant oscillations while the other does not.
This shows a clear separation in the degree distributions of
the generated graphs when different seed sets are used.

\begin{figure}[ht]
    \begin{subfigure}[b]{0.5\textwidth}
        \includegraphics[width=\textwidth]{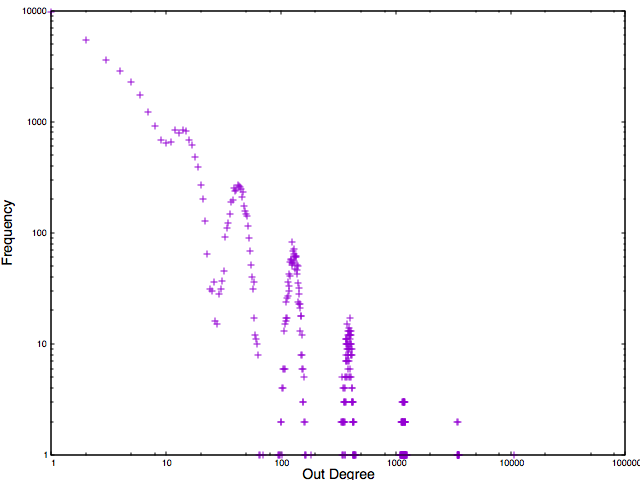}
        \caption{$\rpskg(2^{20}, 16, 0, T)$}
        \label{fig:intro0}
    \end{subfigure}
    \begin{subfigure}[b]{0.5\textwidth}
        \includegraphics[width=\textwidth]{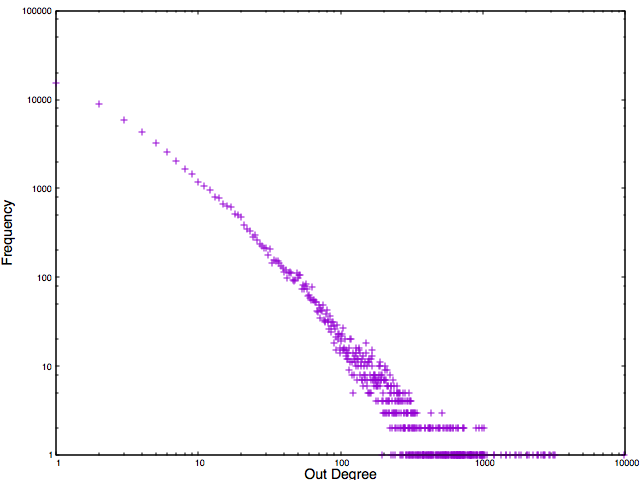}
        \caption{$\rpskg(2^{20}, 15, 1, T)$}
        \label{fig:intro1}
    \end{subfigure}
    \caption{Difference in degree distribution
    when only $2\times 2$ seeds are used (Figure~\ref{fig:intro0})
    versus when
    $3\times 3$ seeds are used with
    $2\times 2$ seeds (Figure~\ref{fig:intro1}). The two algorithms
    belong in two different identifiability classes.
    See Section~\ref{sec:exp} for other simulation results.}
    \label{fig:intro}
\end{figure}

\subsection{Our Contributions}

In our work, we use the statistical and computational identifiability
notions to show separation of SKG models. The separability is
achieved by the following
criteria: the existence of isolated vertices and connected components
in the generated graphs. In this work, we have focused on two
prominent properties of generated graphs but we leave the exploration
of other properties that exhibit separation to future work.

We conjecture that no (randomized) polynomial-time
algorithm can identify a set of seeds (one or more) that will
generate a graph with degree distribution that is
close in statistical distance to
\textit{any} target  distribution.
It is straightforward to see that, in this case,
the Kronecker multiplication primitive acts as a candidate
one-way function.
We also show that this implies that the algorithm cannot possibly be
both SKG identifiable and generate graphs that have no
isolated vertices (in the limit), even if such graphs exists.
An analogous result holds for generating a connected graph.

\begin{conjecture}
There always exists $n\in\naturals$,
a degree distribution $\calD$, and
sequence $\alpha(n)\geq 0$ 
such that
no polynomial-time algorithm $\calA$ 
exists to find a set of seeds $S$ (independent of $n$),
such that
$\tv(\calA(S, n), \calD) \leq \alpha(n)$.
\label{conj:1}
\end{conjecture}

Conjecture~\ref{conj:1} can be used to show the following
theorem:

\begin{theorem}[Informal Version of Theorem~\ref{thm:isolated}]
Let $\calA$ be a polynomial-time algorithm that can generate
a graph, via a seed $S$, which follows a particular
degree distribution for any node set size
$n\in\naturals$.

Then, assuming Conjecture~\ref{conj:1},
there exists a degree distribution $\calD$ and parameterized function
$\alpha(n)\geq 0$ for which
$\calA$ will be unable to ensure that the generated graph
has no isolated vertices while remaining $\alpha(n)$
SKG (statistically) identifiable, in the limit.
\label{thm:intro:isolated}
\end{theorem}

We also show a similar phenomenon for connected components in the
generated graphs of the SKG model assuming Conjecture~\ref{conj:1}:

\begin{theorem}[Informal Version of Theorem~\ref{thm:connected}]
Let $\calA$ be a polynomial-time algorithm that can generate
a graph, via a seed $S$, which follows a particular
degree distribution for any node set size
$n\in\naturals$.

Then, assuming Conjecture~\ref{conj:1},
there exists a degree distribution $\calD$ and parameterized function
$\alpha(n)\geq 0$ for which
$\calA$ will be unable to ensure that the generated graph
always has a connected component while remaining $\alpha(n)$
SKG (statistically) identifiable, in the limit.
\label{thm:intro:connected}
\end{theorem}

We have established, via Theorems~\ref{thm:intro:isolated} 
and~\ref{thm:intro:connected}, the implications of not being able to
identify starting seeds for any given graph.
This is due to the computational
limitations of classical computers so that it is not possible to
efficiently find seeds that can generate a
graph for any given degree distribution that always exhibits certain
properties.
\textit{However, are there even existing efficient algorithms that,
given a seed,
can generate desired distributions without oscillations?}
Algorithm~\ref{alg:rpskg} is one such candidate algorithm.
It is able to mix $2\times 2$ and $3\times 3$ seeds to generate
graphs with degree distributions that do not have undesirable
oscillations. However, with only $3\times 3$ seeds, it still results
in a distribution that is bounded above by an exponential tail. This
result is proved in the following theorem:

\begin{theorem}[Informal Version of Theorem~\ref{thm:oscillations}]
There exists parameters $\alpha, \beta\in\naturals$ such that the
degree of any vertex $v$ in the graph is
$$
\pr[\deg(v) = d] = (1\pm o(1))\frac{\exp(d(2\alpha + \beta) - e^{2\alpha + \beta})}{d!},
$$
where $\deg(v)$ corresponds to the degree of vertex $v$.
\label{thm:intro:oscillations}
\end{theorem}

Note that both $\alpha$ and $\beta$ are discrete
(and not continuous) variables. 
As a consequence, the oscillations occur in the
degree distribution.
For example, when
$|2\alpha + \beta - \ln d| \geq 1/3$, the
probability of having degree $d$ is at most
$\exp(-d/9)$, an exponential tail. But when
$\ln d$ is close to integral, then there will be
many vertices of degree
$d$. Thus, there is a fluctuation in the degree
distribution with exponential tails due to the
discrete nature of the slice parameters,
$\alpha$ and
$\beta$. This reason is
similar as why oscillations occur in the
$2\times 2$ case.

This justifies the use of $2\times 2$ seeds together with
$3\times 3$ seeds in Algorithm~\ref{alg:rpskg}.
Algorithm~\ref{alg:rpskg} is able to solve a degree-4 optimization
problem in order to sample a $3\times 3$ seed from the 
uniform distribution over a given $2\times 2$ seed. 
Prior to our work, it was not known whether such an efficient
procedure exists for the SKG model.

\begin{lemma}[Informal Version of Lemma~\ref{lem:sample3x3}]
There exists a polynomial-time
degree-4 optimization formulation to sample an SKG
$3\times 3$ seed from the uniform distribution over a
$2\times 2$ seed.
\end{lemma}

\begin{lemma}[Informal Version of Lemma~\ref{lem:sample3x1}]
There exists a polynomial-time
degree-2 optimization formulation to sample an SKG
$3\times 1$ seed from the uniform distribution over a
$2\times 1$ seed.
\end{lemma}

Also, as we show,
Algorithm~\ref{alg:rpskg} can also be generalized to sample a
$m\times n$ seeds for any $m, n\in\naturals$.
See Section~\ref{sec:degree-nm} for more details.

The second conjecture states that a $2\times 2$ set of seeds alone or
$3\times 3$ set of seeds alone is not sufficient
to generate graphs with no oscillations in the degree
distribution. (See Section~\ref{sec:exp} for full experimental
validation details.) As can be seen in Figure~\ref{fig:intro},
this conjecture is supported by experiments.
\footnote{See Section~\ref{sec:exp} for more experimental details.}
In particular,
the generation process for
Figure~\ref{fig:intro0} uses only one $2\times 2$ set of seeds while
the generation process for Figure~\ref{fig:intro1} uses
one set of $2\times 2$ seeds and one set of $3\times 3$ seeds.

Conjecture~\ref{conj:1} implies the existence of one-way functions
and cryptographic pseudorandom generators
with seed length
$d(m) < m$.
\footnote{
$G_m:\{0, 1\}^{d(m)}\rightarrow\{0, 1\}^m$ is a 
\textit{cryptographic pseudorandom generator}
where $G_m$ runs in time $m^b$ for some $b > 0$.
$f_n:\{0, 1\}^n\rightarrow\{0, 1\}^n$ is a 
\textit{one-way function} if there is a constant
$b$ such that $f_n$ is computable in time $n^b$ for large enough $n$
and for every constant $c$ and every nonuniform algorithm $\calA$
running in time $n^c$:
$\pr[\calA(f_n(U_n))\in f^{-1}_n(f_n(U_n))]\leq n^{-c}$ for all
sufficiently large $n$ where $U_n$ is a uniform random variable on
$n$ bits.
}
This is because $\calD$ can be chosen to be any distribution, representable with $\poly(n)$-bit strings, and so if the seed $S$
(which could be standardized to be of any length)
cannot be found by polynomial-time algorithms, the
Kronecker multiplication primitive (on the seed) would serve as a 
one-way function for which no efficient algorithms exist to find the
seed.

Now we present the conjecture that is supported by a swath
of experimental details (see Section~\ref{sec:exp}):

\begin{conjecture}

For all $\alpha > 0$ and renormalized
seed sets from either
$\rationals_+^{2\times 2}$ or $\rationals_+^{3\times 3}$,
there always exists a distribution $\calD$ such that
for all algorithms $\calA$,
$\tv(\calA(S, n), \calD) > \alpha$.

\label{conj:2}
\end{conjecture}

In other words, to
generate a graph with degree distribution in \textit{any} possible
identifiability class,
an SKG algorithm cannot do so with only a $2\times 2$ set of seeds or a
$3\times 3$ set of seeds.
This observation is supported by empirical evidence and a new
algorithm that can combine different seed sets.
In Conjecture~\ref{conj:2},
renormalized seed sets refer to sets for which
all entries are divided by the total sum so
that it corresponds to a probability vector.

In Section~\ref{sec:rpskg},
we come up with a new proposal to eliminate
oscillations in any degree distribution generated by stochastic
Kronecker multiplications.  The $\rpskg$ algorithm works by
using a $2\times 2$ seed matrix
and $3\times 3$ seed matrix to generate the graph.
We show that this eliminates oscillations in the degree
distributions (unlike SKG). Also, the Noisy SKG proposal
implicitly combines a $2\times 2$ seed matrix with another seed
matrix to
remove oscillations in the degree distribution
~\citep{Seshadhri:2013}.

In this work, we focus on analyzing properties of SKG-generated
graphs like
connectedness, the existence of isolated vertices, and the
degree distributions.
As we show, there might be phase transitions in SKG
identifiability via a threshold value~\citep{Brennan2019PhaseTF, Ding0XY20, DingWXY21}.
Furthermore, the separations might be able to import
some recent understanding in the theory of
stochastic block models~\citep{AbbeS15, Abbe17,Abbe18}, a study
we leave for future work.

\subsection{Overview of Techniques}

\subsubsection{First and Second Moment Methods}

To prove Theorems~\ref{thm:isolated} and
~\ref{thm:connected}, we construct a sequence of pair of
graphs for which their
degree distributions satisfy a statistical distance upper bound
that is a function of $\alpha(n)\geq 0$. Then we consider the
event that a vertex is isolated. We prove that, over the entire graph,
for one sequence of graphs, 
$\pr[\exists\,\,\text{isolated nodes}]$ goes to 0 and for the other 
sequence, it goes to 1. 
For one sequence of graphs, we only consider the first moment of the
random variable that counts the number of isolated vertices. For the
second sequence of graphs, we consider the second moment of the
random variable.

\subsubsection{The \rpskg algorithm}

Algorithm~\ref{alg:rpskg} is able to generate a graph without
oscillations (as detailed in Section~\ref{sec:exp}).
The algorithm mixes the use of $2\times 2$ with $3\times 3$ seeds.
A subroutine samples a $3\times 3$ seed from the uniform
distribution over the $2\times 2$ seed. We show that this
sampling problem corresponds to solving a degree-4 optimization problem
in the variables of the $2\times 2$ seed.
Also, even though Algorithm~\ref{alg:rpskg} is written in a serial
form for clarify sake, it is easy to see that it is
easily parallelizable.

\subsection{Why is a New Approach Necessary and Important?}

Graph analysis remains important and relevant.
On one hand, we need to design efficient algorithms to infer properties
from these large ``real-world'' graphs. As a result, a
lot of research has been done to enable generation of graph models that realistically model
``real-world'' graphs.
The celebrated work of Erd\H{o}s and R\'{e}nyi
~\cite{ErdosRenyi:60, Gilbert59}
for fast construction of random graphs just by
inserting each edge with probability $p > 0$ does not seem to be a good solution in this context,
since social graphs have many properties (for example triadic closure and homophily) that
random graphs don’t model well. Thus, a more
realistic model is needed. The model that is
usually used in practice
is the Stochastic Kronecker Model (SKG) proposed in 2010 by Leskovec et al.~\cite{Leskovec:2010} (further
discussed in Section~\ref{sec:skg}).
This model has been chosen to create graphs for the
Graph500 supercomputer benchmark
~\cite{Graph500}. SKG models are used for
various reasons: it has very few parameters and
can generate graphs fully in parallel --- on a
per-edge basis.

However, the SKG model has several
problems that differentiate
its generated graphs from real network structures
(social networks, for instance).
In particular, Seshadhri et al.
~\cite{Seshadhri:2013} proved that SKG models
\textit{cannot} generate a power law or
log-normal distribution. Graphs generated by the
Stochastic Kronecker Graph model
is most accurately characterized as fluctuating
between a log-normal distribution and an exponential
tail. Figure~\ref{fig:intro0} is an example of
a plot (on a log-log scale) of the degree
distribution of a graph generated by SKG. This
figure clearly shows unwanted oscillations in the
degree distribution.
The Noisy Stochastic Kronecker Graph (NSKG)
model attempts to fix some of the problems
with the SKG model. The
main problem NSKG fixes is that it gets rid of
the oscillations in the degree distributions
of generated graphs. However, NSKG generates multiple
random numbers that
scales proportionally (logarithmically)
with the
number of edges generated. These random numbers are
additional parameter artifacts
that must be recorded for reproducibility sake, thus
making NSKG unwieldy.

In this paper,
we present a new model, the
Relative Prime Stochastic Kronecker Graph (\rpskg)
model, that can generate graphs
with log-normal degree distributions that do not
have unwanted oscillations (unlike SKG) while adding no
additional parameters (unlike NSKG).
In addition, we present theoretical and experimental
results about the \rpskg model that show its
superiority to SKG.

The \textit{main}
contribution of our paper, though, is in presenting
criteria (via definitions of statistical and computational
identifiability) to classify and separate SKG models.
The implications of the conjectures show
the necessity of the classification of SKG models.
The \rpskg model is a family of algorithms that is adaptive to
interpolation of seeds of different dimensions. We believe this
model and the $\rpskg$ algorithm warrants further study.

\subsection{Additional Notation and Background}
\label{sec:additionalback}

In Section~\ref{sec:related},
we dive into contextual work related to our research
but first we introduce some preliminary notation that
will be useful
throughout this paper.

The SKG, NSKG, and \rpskg models all belong to a family
on Kronecker models that rely on a non-traditional
matrix operation, the Kronecker product
(otherwise known as the Tensor product).
The Kronecker product of two matrices
$A \in \mathbb{R}^{m \times n}$, $B \in \mathbb{R}^{p \times q}$ is an $mp \times nq$ matrix, defined as follows:

\[A \otimes B = 
  \begin{bmatrix}
    a_{11}B & \cdots & a_{1n}B\\
    \vdots & \ddots & \vdots\\
    a_{m1}B & \cdots & a_{mn}B\\
  \end{bmatrix}
  = \begin{bmatrix}
    a_{11}b_{11} & \cdots & a_{1n}b_{1q}\\
    \vdots & \ddots & \vdots\\
    a_{m1}b_{p1} & \cdots & a_{mn}b_{pq}\\
  \end{bmatrix}
\]

Also common to the three models is a
\textit{initiator matrix}, which is usually
a $2\times 2$ matrix. 
The Kronecker power of this matrix is then
computed until we obtain a (usually large) matrix
with dimensions of an adjacency matrix of the
the graph we wish to generate. The generated matrix
is a probability matrix (all values sum to one if the matrices involved in the product also have values that sum to one) which
we can sample from to generate edges.
Thus, we use this matrix to generate edges for a graph.

In practice, we \textit{do not} actually generate
this large matrix but simulate its generation for use
in sampling on a per-edge basis.
In addition to the initiator matrix, the models
require a parameter for the number of
Kronecker products performed. In both NSKG and SKG,
$\ell$ is the
number of Kronecker products.
As a consequence, the number of
nodes generated depends on
$\ell$. For both NSKG and SKG,
$2^\ell$ nodes are generated. On the other hand,
for \rpskg, the number of nodes generated is
$3^k\cdot 2^\ell$ where
$k$ is the number of $3\times 3$ seeds used.
Refer to Algorithm~\ref{alg:rpskg} for a
description of the use of the seeds.
Note that, via standard rounding techniques, the generation
process can be adjusted to generate any number of nodes.

The SKG is not only a theoretical model; it is heavily used in practice as well. For example, the Graph500 supercomputer benchmark uses SKG with the
following initiator matrix

\begin{equation}
\label{eq:Graph500}
T = \begin{bmatrix}
    t_1 & t_2\\
    t_2 & t_3\\
  \end{bmatrix}
  = \frac{1}{16}
  \begin{bmatrix}
    9 & 3\\
    3 & 1\\
  \end{bmatrix}
\end{equation}

\noindent and generates graphs with $m = 16\cdot 2^\ell$
edges, $n=2^\ell$ nodes where
$\ell\in\{26, 29, 32, 36, 39, 42\}$.

Notice that the Graph500
initiator matrix in Equation~\ref{eq:Graph500}
is symmetric. Often times, in other situations,
this matrix is symmetric. We exploit this
property in our theoretical analysis. We now precisely define
the stochastic Kronecker graph model:

\begin{definition}[Stochastic Kronecker Graph]
A \textbf{Stochastic Kronecker Graph} (SKG) is defined by
an integer $\ell$ and a matrix
$S\in\reals^{s\times s}$ for some $s\in\naturals$. 
$S$ is the base or initiator matrix.

The graph generated will have $n = s^\ell$ vertices, where
each vertex is labeled by a
unique bit vector of length $\ell$. Given two
vertices $u$ (with label $u_1u_2\ldots u_\ell$) and
$v$ (with label $v_1v_2\ldots v_\ell$) the probability of
the edge $(u, v)$ appearing in the resulting graph is
$\prod_i^\ell S[u_i, v_i]$, independent of the presence of other
edges.
\label{def:skg}
\end{definition}

Often, the Stochastic Kronecker Graph definition (Definition~\ref{def:skg}) is assumed to rely on 
$2\times 2$ seed matrices to generate the graph. We generalize this
notion to allow for seeds of any dimension. Also, although the
number of edges in the generated graph is often variable, we can
also impose a restriction on the number of edges. e.g.,
by sampling from a truncated distribution like the truncated normal
and using rounding techniques.
As we will discuss, Definition~\ref{def:skg} is inspired by the
Recursive Matrix model~\citep{Groer:2011}.

\section{Related Work}
\label{sec:related}

In this section, we survey some graph generation
models and will focus on
discussing SKG, NSKG,
and Multiplicative Attribute Graph models. Also,
we discuss some of the work
done in the validation of graph generation models.
See works of Margo~\cite{Margo17} for
a more complete survey. 

\subsection{Stochastic Kronecker Graphs}
\label{sec:skg}

In a seminal paper, Leskovec et al.~\cite{Leskovec:2010} define the Stochastic Kronecker Graph (SKG) model
that can be very effective in capturing the characteristics of real social networks (for example,
small diameters and heavy-tailed degree distributions).

The SKG model is a generalization of the Recursive Matrix model introduced by Chakrabarti et al.~\cite{Chakrabarti:2004}.
To generate a Stochastic Kronecker graph, they use a single $K\times K$ initiator matrix $T$.
Often $K = 2$, a model equivalent to R-MAT~\cite{Groer:2011}. 
Each of the $K^2$
entries in the initiator matrix
represents the probability of generating an edge for a specific source and target vertex regions
of the graph. For example, for
the Graph500 benchmark $2\times 2$ seed matrix
(Equation~\ref{eq:Graph500}), the seed values mean that
there is roughly a $56\%$ chance of inserting an edge in the top-left quadrant of the generated
graph.
The resulting graph can be naively generated by the Kronecker product as follows:

$$
P = T^{\otimes\ell} = T\otimes T\otimes\ldots\otimes T
= T^{\otimes\ell - 1}\otimes T
$$

In the resulting matrix, $P_{ij}$ represents the probability that there is an edge connecting the
nodes $i$ and $j$. In practice, the Kronecker is simulated by recursively choosing a sub-region of
matrix $T^{\otimes r}$
(after $r\leq\ell$ steps) until we descend on a single cell of $P$ and then place an edge.
To further illustrate the generation process, suppose we use a $2\times 2$ initiator matrix and wish
to generate graph on $m = 16\cdot 2^\ell$
(often used in practice---for Graph500 for example). Then
for every edge, we divide the adjacency matrix into four quadrants, and choose one of them
with the corresponding probability that the
$2\times 2$ seed matrix suggests. Once a quadrant is chosen,
this procedure is repeated recursively. Intuitively, if we consider the indices of the vertices to
be binary numbers, the $r^{th}$ step of this process fixes the $r$th most important bit of the two
endpoints of the edge.
Note that, using this approach some of the edges may be inserted multiple times. According
to experiments, however, this occurs rarely --- less that 1\% of the total number of edges.
One of the most important characteristics of this model, that makes it widely applicable in
practice, is that it can generate graphs fully in parallel. Because of this property, the Graph500
supercomputer benchmark uses SKG by default to generate large graphs.
Although the SKG model has been influential and is widely used in practice, there are
several problems with the model:
\begin{enumerate}
\item The degree distribution of the generated graph does not have the wanted heavy-tail
behavior but there are large oscillations instead. First, Gr\H{o}er et al.~\cite{Groer:2011}
proved that the
degree distribution behaves like a sum of Gaussians and later Seshadhri et al.~\cite{Seshadhri:2013} showed
that it oscillates between a log-normal distribution and an exponential tail.
\item A very large fraction of the vertices of every graph is isolated and thus the graphs created
with this model have a much smaller size and are much more dense than expected.
\item The max core numbers in the generated graphs are extremely small --- much smaller than those
corresponding of real graphs.
\end{enumerate}

Taking these problems into consideration, one would like to improve this model in order to
keep the easy generation and parallelization but make behave it more like real-life graphs. In
Section \ref{sec:nskg}, we briefly discuss the fix that was proposed by Seshadhri et al.~\cite{Seshadhri:2013} and
note some potential problems. In Section
~\ref{sec:rpskg},
we come up with a new proposal to eliminate the
oscillations in the degree distribution and
subsequently validate our proposal theoretically
and experimentally.

Mahdian et al.~\cite{Mahdian:2007} show
necessary and sufficient
conditions for stochastic Kronecker graphs to be connected or to have giant components of size $\Theta(n)$ with
high probability. They study the connectivity
and searchability
probabilities of both random graphs
(a generalization of the Erd\H{o}s-R\'{e}nyi model)
and stochastic Kronecker graphs.

Mihyun Kang et al.~\cite{KangKochMakai:2015}
also examine properties of stochastic
Kronecker graphs but focus more on
on showing that the graphs do not feature
a power law degree distribution nor a log-normal
degree distribution.

\subsection{Noisy Stochastic Kronecker Graphs}
\label{sec:nskg}

Seshadhri et al.~\cite{Seshadhri:2013}
provide rigorous proofs for many of the problems
with the SKG model (stated previously
in Section~\ref{sec:skg}).
In addition, Seshadhri et al.
provide a new model --- the Noisy SKG ---
that fixes some
of the problems with the original SKG model. The idea is to add noise in every step of
the computation (i.e., of the edge creation) so that edges are not sampled from the actual
distribution created by the initial
$2\times 2$ seed matrix,
but rather from a noisy version of it. At each
step $r\leq \ell$ of the recursive generation of the Kronecker product, $\mu_r$ is the noise factor chosen
from the uniform distribution with range $[-b, +b]$
where $b \leq \min(t_2,(t_1 + t_4)/2)$
is the noise
parameter. The initiator matrix $T$ is transformed on a per-level basis to become:

\[
 \begin{bmatrix}
    t_1 - \frac{2\mu_rt_1}{t_1 + t_4} &
    t_2 + \mu_r\\
    t_3 + \mu_r &
    t_4 - \frac{2\mu_rt_4}{t_1 + t_4}\\
  \end{bmatrix}
\]

Although, the generated graph of NSKG models
are an improvement over SKG models, there are
still some problems with the NSKG model:
\begin{enumerate}
\item Adding too much noise could
destroy some ``real-world'' properties of the graph. In a sense, adding noise is
equivalent to superimposing an Erd\H{o}s-R\'{e}nyi graph on the Kronecker graph. The larger the
noise, the quicker the generated graph converges to a random graph.
\item During and after graph generation, we have to keep track of the $\ell$ random $\mu_r$
$(1\leq r\leq \ell)$
used for noise thus adding to our already cumbersome parameter list.
\item The large amount of noise added could sometimes make the generated graphs significantly
different from others generated using the same initiator matrix. Therefore, to make safe
conclusions about experiments based on the generated graphs, users must generate many
such graphs thereby adding to the computation cost (e.g., supercomputing cost for Graph500).
\item Related to our previous point, since the noise is white (i.e., its expectation is 0) if we
create several graphs, on average we expect that our results would
not be noisy, and thus the
oscillations in the degree distribution will reappear. This means that although the degree
distribution of every graph is log-normal, the degree distribution of the whole family of
graphs is still oscillating between a log-normal and an exponential tail.
\end{enumerate}

We would like to note that although the NSKG model was communicated to the Graph500
committee, it is not used in the standard and is in fact compile-time disabled by default. This might
imply that the Graph500 committee had some concerns (probably similar to ours) about the
NSKG model.

\subsection{Multiplicative Attribute Graph Model}

Kim and Leskovec~\citep{KimL12} 
present the \textit{Multiplicative Attribute Graph} (MAG)
model which
is very similar to SKG. MAG models are essentially a generalization of SKG where
each level (for generating a single edge)
may have a different matrix $T$
(denoted $\Theta$ in their paper).
The $\Theta$ matrices are the link-affinity
matrices that be used to model link structure
between nodes.
Kim and Leskovec
show that certain configurations of these matrices
can lead to power-law or log-normal distributions.

However, MAG models are more complex than the family
of Stochastic Kronecker Graph models
as each node stores attributes and
the different link-affinity matrices add complexity
to the model specification.

\subsection{Validation of Degree Distributions}

We agree with
Mitzenmacher ~\cite{Mitzenmacher:2005} with the sentiment that in studying
power-law (and related) distributions of graph models we should
not only \textit{observe}, \textit{interpret}, and
\textit{model} but also aim to \textit{validate} and
\textit{control} our research.
As a consequence, we do not just present our algorithm but show
through experiments (in Section~\ref{sec:exp})
that graphs generated by
\rpskg approximately follow a log-normal distribution.

According to Newman~\cite{Newman:2005}, distributions
that follow log-normal typically arise when multiplying
random numbers. Mitzenmacher~\cite{Mitzenmacher:2003}
uses the term
\textit{multiplicative process} to describe this
model. The
log of the product of a large number of random numbers is
the sum of the logarithms of those same random numbers,
and by the central limit theorem such sums have a normal
distribution essentially regardless of the distribution of
the individual numbers.

Furthermore, Mitzenmacher describes some
generative models for power law and log-normal
distributions~\cite{Mitzenmacher:2003}.
Specifically, he describes how the preferential
attachment models leads to a power law distribution.
In preferential attachment models, new objects
tend to attach to popular objects. In the case of
the Web graph, new links tend to go to pages that
already have links.

\subsection{Relation to the Chung-Lu Model}

We also note that the SKG model is related to other graph models. We
cannot possibly exhaustively relate to every other possible graph
generation model but will focus on one other:
the Chung-Lu model. In~\citep{SeshadhriPK12},
Seshadhri, Pinar, and Kolda explore similaries between SKG models and
the Chung-Lu model.

Let $d_1, d_2, \ldots, d_n$ be a sequence of $n$ in-degrees and
$d_1', d_2', \ldots, d_n'$ be a sequence of $n$ out-degrees such that
$\sum d_i = \sum d_i' = m$.  Then consider the probability matrix
$P_{CL}$ for
making $m$ edge insertions:
the $(i, j)$th entry is $d_i d_j'/m^2$.
In the Chung-Lu model, $P_{CL}$ is used to generate the graph and
can be used to model any degree distribution~\citep{ChungLu02}.

\section{Existence of Isolated Vertices and Connected Components}
\label{sec:isolated}

Conjecture~\ref{conj:1} separates statistical and computational
SKG identifiability. In particular, as we will show, if the
(poly-time) algorithm cannot find a seed that will follow
a specific degree distribution, the algorithm cannot ensure
certain properties of the graph. Specifically,
even if a graph---which follows
a certain degree distribution up to statistical distance---
without any isolated
vertices exists, the algorithm will be unable to generate such
a graph.

There could be other ramifications of the separations of statistical
and computational identifiability in terms of graph properties. We first
consider a most obvious one --- existence of isolated vertices.

\subsection{Isolated Vertices}

Our work shows a threshold phenomena for the existence of
isolated nodes. The function $\alpha(n)\geq 0$ can be thought of
as one with arbitrarily slow growth to infinity. For example,
$\alpha(n) = \log \log \log n$.

\begin{theorem}

Let $\calA$ be a polynomial-time algorithm that can generate
a graph, via a seed $S$, which follows a particular
degree distribution for any node set size
$n\in\naturals$.

Then assuming Conjecture~\ref{conj:1},
there exists a degree distribution $\calD$ and
$\beta(n)\geq 0$ for which
$\calA$ will be unable to ensure that the generated graph
has no isolated vertices while remaining $\beta(n)$ SKG statistically
identifiable, in the limit.

\label{thm:isolated}
\end{theorem}

\begin{proof}

Without loss of generality, we will consider the uniform degree
distribution since for any random variable $X$ with CDF $F_X$,
we can transform to a uniformly distributed random variable
via $F_X(X)$.

Consider the uniform degree distribution generated by a 
homogeneous Bernoulli graph: each graph sequence $G_n$
has probability of $p_n$ of any two nodes having an edge.
Then by Lemma~\ref{lem:isolated}, if $p_n$ is above some 
threshold, 
$\pr(\exists\text{ isolated nodes})\rightarrow 0$.
If $p_n$ is below some threshold,
$\pr(\exists\text{ isolated nodes})\rightarrow 1$.

By the conjecture, since the algorithm is bounded by
a polynomial-time runtime, it cannot distinguish between
two seeds $T_1$ and $T_2$, exactly
one of which generates a graph without isolated nodes
(in the limit). Moreover, the edge link probabilities of
the two graphs are chosen to be within $2\alpha(n)/n$ to
induce closeness in degree distributions.
In particular, the TV distance between the distributions of the two
generated graphs is at most $\beta(n) \leq \alpha(n)$.

As a result, by Lemma~\ref{lem:isolated}, the algorithm $\calA$
cannot ensure the generated graph is SKG statistically 
identifiable and has no isolated vertices.

\end{proof}

\begin{lemma}

Fix $s\in\naturals$ and let $T_1, T_2\in\reals^{s\times s}$ be two SKG seeds.
Let $\alpha(n)\geq 0$ be a sequence that
depends on $n$.
Then let $T_1^{\otimes \ell}$ produce a $n$-node graph $G_1$ with
(homogeneous) link probability of 
$p_n \geq \frac{\log n + \alpha(n)}{n}$.
And let $T_2^{\otimes \ell}$ produce a $n$-node graph $G_2$ with
(homogeneous) link probability of 
$p_n \leq \frac{\log n - \alpha(n)}{n}$.

Then
$$
\pr[\exists\text{ isolated nodes in }G_1]\rightarrow 0,
$$
$$
\pr[\exists\text{ isolated nodes in }G_2]\rightarrow 1.
$$

Furthermore,
$$
\E[\text{isolated nodes in }G_1] \leq e^{p_n}e^{-\alpha(n)}.
$$
\label{lem:isolated}
\end{lemma}

\begin{proof}

Consider the $n$ nodes of a graph $G_n$:
$v_1, v_2, \ldots, v_n$.
Let $E_i$ denote the event that any vertex $v_i$ is isolated.
The the number of isolated vertices in $G_n$ is
$$
X_n = \sum_{i=1}^n\ind[E_i].
$$
The probability of the event $E_i$ is
$(1-p_n)^{n-1}$ (using independence and
since w.p. at least $1-p_n$, any vertex is
not connected to any others).
Then
\begin{equation}
\E[X_n] = 
\sum_{i=1}^n\E[\ind[E_i]]
= \sum_{i=1}^n\pr[E_i]
= n(1-p_n)^{n-1}.
\label{eq:linearity}
\end{equation}

Using the taylor expansion of
$\log(1-t) = -\sum_{k=1}^\infty\frac{t^k}{k}$, it can be
derived that
\begin{equation}
\log(1-t) \leq -t,\quad\forall t\in[0, 1).
\label{eq:log1}
\end{equation}
Furthermore,
\begin{equation}
\log(1-t) \geq -t-t^2,\quad\forall t\in[0, 1/2].
\label{eq:log2}
\end{equation}

Case (i):
For graphs of type $G_1$, we have
$p_n \geq\frac{\log n + \alpha(n)}{n}$ for some
$\alpha(n)\rightarrow\infty$. Then by 
Equation~\ref{eq:linearity} and
Equation~\ref{eq:log1}:
\begin{align}
\E[X_n] &= e^{\log n + (n-1)\log(1-p_n)} \\
&\leq e^{\log n - (n-1)p_n} \\
&= e^{p_n}e^{\log n - np_n} \\
&\leq e^{p_n}e^{-\alpha(n)}.
\end{align}

So,
$\E[\text{isolated nodes in }G_1] \leq e^{p_n}e^{-\alpha(n)}$
and by Markov's inequality,
$$
\pr[\exists\text{ isolated nodes in }G_1] = \pr(X_n
\geq 1) \leq \E[X_n]\rightarrow 0.
$$

Case (ii):
We will consider graphs of type $G_2$ and apply Chebyshev's
inequality as follows:
\begin{align}
\pr(X_n = 0) &\leq \pr(X_n \leq 0) \\
&\leq \pr(|X_n - \E[X_n]|\geq\E[X_n]) \leq \frac{\Var(X_n)}{(\E[X_n])^2}.
\label{eq:cheby}
\end{align}

For graphs of type $G_2$,
$p_n\leq \frac{\log n - \alpha(n)}{n}$ for some 
$\alpha(n)\rightarrow\infty$. Clearly, for large $n$,
$p_n\leq 1/2$. Then by
Equation~\ref{eq:linearity} and
Equation~\ref{eq:log2}:
\begin{align}
\E[X_n] &= e^{\log n + (n-1)\log(1-p_n)} \\
&\geq e^{\log n - (n-1)(p_n + p_n^2)} \\
&\geq e^{\log n - n p_n - n p_n^2}.
\end{align}

Since $n p_n \leq \log n - \alpha(n)$ and
$n p_n^2 = 1/n \cdot (np_n)^2 \leq 1/n\cdot(\log n)^2\leq 1$,
then as $n\rightarrow\infty$,
$$
\E[X_n]\geq e^{\alpha(n) - 1}\rightarrow\infty.
$$

Next we analyze the variance:
$$
\E[X_n^2] = 
\E\left[\sum_{i=1}^n\sum_{j=1}^n\ind[E_i]\ind[E_j]\right]
= \sum_{i=1}^n\sum_{j=1}^n\pr(E_i, E_j).
$$

Then
$$
\Var(X_n) = \E[X_n^2] - (\E[X_n])^2
= \sum_{i=1}^n\sum_{j=1}^n\left(\pr(E_i, E_j) - \pr(E_i)\pr(E_j)\right).
$$
Also,
$$
\pr(E_i, E_j) = \pr(E_i)\pr(E_j\,\mid\,E_i)
= (1-p_n)^{n-1}(1-p_n)^{n-2}
= (1-p_n)^{2n-3},
$$

since for any two notes $v_i\neq v_j$,
the conditional probability of $v_j$ being isolated, given
that $v_i$ is isolated, is $(1-p_n)^{n-2}$.

Then since $\pr(E_i) = (1-p_n)^{n-1}$, we obtain that
\begin{align}
\Var(X_n) &= n(\pr(E_1)-P(E_1)^2) + n(n-1)(\pr(E_1, E_2) - \pr(E_1)^2) \\
&\leq n\pr(E_1) + n^2(\pr(E_1, E_2) - \pr(E_1)^2) \\
&= n(1-p_n)^{n-1} + n^2((1-p_n)^{2n-3} - (1-p_n)^{2n-2}) \\
&= n(1-p_n)^{n-1} + n^2(1-p_n)^{2n-2}\frac{p_n}{1-p_n} \\
&= \E[X_n] + \E[X_n]^2\frac{p_n}{1-p_n}.
\end{align}

So by Equation~\ref{eq:cheby},
$$
\pr(X_n = 0) \leq \frac{\Var(X_n)}{(\E[X_n])^2}
\leq \frac{1}{\E[X_n]} + \frac{p_n}{1-p_n}
\rightarrow 0.
$$

This implies that 
$\pr[\exists\text{ isolated nodes in }G_2]\rightarrow 1.$

\end{proof}

\subsection{Connected Components}

We prove a similar result (to Lemma~\ref{lem:isolated}) but for
connected components. In fact, the proof of Lemma~\ref{lem:connected}
can be seen as an extended
form of the proof of Lemma~\ref{lem:isolated}.
First, by Lemma~\ref{lem:isolated}, the
existence of disconnected components implies that the graph cannot
be connected. Of course, having no isolated vertices
does not necessarily imply that the graph is connected. 
The probability that a graph is disconnected is upper bounded by the
probability that no components of size $k$, where $1\leq k\leq n/2$,
exist in the graph. We
will show that the latter probability goes to 0 via the use of
Cayley's theorem (Theorem~\ref{thm:cayley}). 
(We only need to consider components of size $1\leq k\leq n/2$ since
the graph on $n$ nodes is disconnected if and only if there exists
a connected component of size at most $n/2$.)

\begin{lemma}

Fix $s\in\naturals$ and let $T_1, T_2\in\reals^{s\times s}$ be two SKG seeds.
Then let $T_1^{\otimes \ell}$ produce a $n$-node graph $G_1$ with
(homogeneous) link probability of 
$p_n \geq \max(9e/n, \frac{\log n + \alpha(n)}{n})$.
And let $T_2^{\otimes \ell}$ produce a $n$-node graph $G_2$ with
(homogeneous) link probability of 
$p_n \leq \frac{\log n - \alpha(n)}{n}$.

Then
$$
\pr[G_1\text{ is connected}]\rightarrow 1,
$$
$$
\pr[G_2\text{ is connected}]\rightarrow 0.
$$

\label{lem:connected}
\end{lemma}

\begin{proof}

First consider graphs of type $G_2$ where
$p_n \leq \frac{\log n - \alpha(n)}{n}$.
Then,
$$
\pr[G_2\text{ is connected}] \rightarrow 0,
$$
since by Lemma~\ref{lem:isolated},
$\pr[\exists\text{ isolated nodes in }G_2]\rightarrow 1$ so that
$$\pr[G_2\text{ is connected}] \leq 1-\pr[\exists\text{ isolated nodes in }G_2]\rightarrow 0.$$

For the rest of the proof, we will focus on graphs of type $G_1$
where $p_n \geq \max(9e/n, \frac{\log n + \alpha(n)}{n})$.

Consider the $n$ nodes of a graph $G_1$:
$v_1, v_2, \ldots, v_n$. Let
$X^{(k)}$ be the number of components of size $k$ in graph $G_1$.
By Markov's inequality,
\begin{align}
\pr[G_1\text{ is disconnected}]
&\leq \pr\left(\sum_{k\leq n/2} X^{(k)}\geq 1\right) \\
&= \sum_{k\leq n/2} \E[X^{(k)}].
\label{eq:g1-discon}
\end{align}

Let us start off simple: 
$X^{(1)}$ is the number of isolated vertices
and we have that, by Lemma~\ref{lem:isolated},
\begin{equation}
\E[\text{isolated nodes in }G_1] \leq e^{p_n}e^{-\alpha(n)}\rightarrow 0.
\end{equation}
So $\E[X^{(1)}]\rightarrow 0$.

Now consider $X^{(2)}$ which is
$$
\E[X^{(2)}] = {n\choose 2}p_n(1-p_n)^{2(n-2)},
$$
since every node pair is linked with probability $p_n$ and
each node pair is not connected to any other node in the graph
with probability $(1-p_n)^{2(n-2)}$. Then it follows that
\begin{align}
\E[X^{(2)}] 
&\leq n^2p_ne^{-2(n-2)p_n} \\
&= p_ne^{4p_n}e^{-2(np_n-\log n)} \\
&\leq p_ne^{4p_n}e^{-2\alpha(n)} \\
&\leq e^{4-2\alpha(n)} \rightarrow 0,
\end{align}
where we used that ${n\choose 2} \leq n^2$ for $n\geq 1$,
$1-t\leq e^{-t}$ for all $t\in[0, 1)$,
$np_n - \log n \leq \alpha(n)$,
and $p_n\leq 1$.

Next, we consider $X^{(k)}$ for all integer $k\geq 3$.
A (trivial) fact we can use is that any connected component in a graph
contains a spanning tree and consider such spanning trees.
Then since
$X^{(k)} = \sum_{C}\ind[C\text{ is a connected component of size }k] \leq \sum_{C}\sum_{T}\ind[T\text{ is a subgraph of }G_1]\ind[C\text{ is isolated}]$, we have that
\begin{equation}
\E[X^{(k)}]
\leq \sum_{C}\sum_{T}\pr[T\text{ is a subgraph of }G_1]\pr[C\text{ is isolated}].
\label{eq:sub-C}
\end{equation}

We know that any tree $T$ with $k$ vertices has \textit{exactly}
$k-1$ edges. Also, any component of size $k$ and its complement
has \textit{exactly} $k(n-k)$ node pairs between them. Then
$$
\pr[C\text{ is isolated}] = (1-p_n)^{k(n-k)},
$$
and
$$
\pr[T\text{ is a subgraph of }G_1] = p_n^{k-1}.
$$

Then Equation~\ref{eq:sub-C} becomes
\begin{equation}
\E[X^{(k)}] \leq {n\choose k}k^{k-2}p_n^{k-1}(1-p_n)^{k(n-k)},
\label{eq:fact-C}
\end{equation}
since by Cayley's Theorem (Theorem~\ref{thm:cayley}), the number of
trees on a set of $k$ vertices is $k^{k-2}$.

By applying the Stirling approximation lower bound
$k! \geq e^{-k}k^{k+1/2}\sqrt{2\pi}$, we have
\begin{align}
{n\choose k}
&= \frac{n(n-1)\cdots(n-k+1)}{k!} \\
&\leq \frac{n^k}{k!} \\
&\leq \frac{n^ke^k}{k^k}.
\end{align}
Also for $n\geq 2$, $k\geq n/2$ we have
$$
(1-p_n)^{k(n-k)} \leq e^{-k(n-k)p_n} \leq e^{-knp_n/2},
$$
so that Equation~\ref{eq:fact-C} becomes
$$
\E[X^{(k)}] \leq \frac{n^ke^k}{k^2}p_n^{k-1}e^{-knp_n/2}
= \frac{1}{p_nk^2}(enp_ne^{-np_n/2})^k
= \frac{1}{p_nk^2}q_n^k,
$$
where $q_n = enp_ne^{-np_n/2} \leq 8e(n p_n)^{-1}$ since
$e^t\geq 1/2\cdot t^2$.
And by the assumption in the theorem statement, $n p_n\geq 9e$ so that
$q_n \leq \frac{8}{9} < 1$.

As a result,
\begin{align}
\sum_{3\leq k\leq n/2}\E[X^{(k)}] 
&\leq \sum_{3\leq k\leq n/2} \frac{1}{p_nk^2}q_n^k, \\
&\leq \frac{1}{9p_n}\sum_{k=3}^\infty q_n^k \leq \frac{1}{9p_n}\cdot\frac{q_n^3}{1-q_n} \\
&\leq \frac{1}{p_n} q_n^3 \\
&= e^{3+\log n + 2\log(np_n)-3/2np_n} \\
&\leq e^{3+\log n - n p_n} \\
&\leq e^{3-\alpha(n)},
\end{align}
since (i) $np_n - \log n \geq \alpha(n)$, (ii)
$e^t \geq t + 1/2\cdot t^2 + 1/6\cdot t^3 = 4t + (t-3)(t+6)\geq 4t$
for $t\geq 3$ implies that
$\log(n p_n) \leq 1/4 n p_n$ for $n p_n\geq e^3$ (holds since
$np_n\geq 9e$).

Finally, Equation~\ref{eq:g1-discon} becomes
\begin{align}
\pr[G_1\text{ is disconnected}]
&\leq \sum_{k\leq n/2} \E[X^{(k)}] \\
&= \E[X^{(1)}] + \E[X^{(2)}] + \sum_{3\leq k\leq n/2}\E[X^{(k)}] \\
&\rightarrow 0,
\end{align}
as $n\rightarrow\infty$.

This completes the proof.

\end{proof}

\begin{theorem}[Cayley's Theorem/Formula. See~\citep{cayley_2009}]
The number of spanning trees of a complete graph on $n$ vertices is
$n^{n-2}$.
\label{thm:cayley}
\end{theorem}

Now we can obtain the following theorem from
Lemma~\ref{lem:connected}:

\begin{theorem}

Let $\calA$ be a polynomial-time algorithm that can generate
a graph, via a seed $S$, which follows a particular
degree distribution for any node set size
$n\in\naturals$.

Then assuming Conjecture~\ref{conj:1},
there exists a degree distribution $\calD$ and
$\beta(n)\geq 0$ for which
$\calA$ will be unable to ensure that the generated graph
is connected while remaining $\beta(n)$ SKG statistically
identifiable, in the limit.

\label{thm:connected}
\end{theorem}

\begin{proof}

Again, let us  consider the uniform degree
distribution since for any random variable $X$ with CDF $F_X$,
we can transform to a uniformly distributed random variable
via $F_X(X)$.

We can focus on the uniform degree distribution generated by a 
homogeneous Bernoulli graph: each graph sequence $G_n$
has probability of $p_n$ of any two nodes having an edge.
Then by Lemma~\ref{lem:connected}, if $p_n$ is above some 
threshold, $\pr[G_n\text{ is connected }]\rightarrow 1$.
Otherwise, $\pr[G_n\text{ is connected }]\rightarrow 0$.

By the conjecture, since the algorithm is bounded by
a polynomial-time runtime, it cannot distinguish between
two seeds $T_1$ and $T_2$, exactly
one of which generates a graph that is connected
(in the limit). Moreover, the edge link probabilities of
the two graphs are chosen to be within $2\alpha(n)/n$ to
induce closeness in degree distributions.
In particular, the TV distance between the distributions of the two
generated graphs is at most $\beta(n) \leq \alpha(n)$.

As a result, by Lemma~\ref{lem:connected}, the algorithm $\calA$
cannot ensure the generated graph is SKG statistically 
identifiable is connected.

\end{proof}

Thus, if the conjecture is true, we cannot expect 
SKG algorithms to follow certain degree distributions and
ensure no isolated vertices or disconnected components.

Next, we turn to the $\rpskg$ algorithm that can interpolate
seeds to generate graphs that follow degree distributions.
In practice (via experiments), the algorithm yields superior
improvements over existing SKG algorithms (NSKG and
vanilla SKG).

\section{Relative Prime Stochastic Kronecker Graphs}
\label{sec:rpskg}

The SKG model~\citep{Leskovec:2010} has been effective in
capturing degree characteristics of heavy-tailed degree
distributions. But it can result in significant oscillations
in the degree distribution of the graphs produced.
On the other hand,
the per-level noise multipliers added in the NSKG model can
destroy noticeable properties of the graph.
The NSKG model also comes with significant additional
parameterization as one might have to keep track of the randomness
used to generate the perturbed seed in each level.

The $\rpskg$ algorithm aims to fix problems in the generation
procedure based on one starting seed. The algorithm proceeds by 
interpolating the use of a $2\times 2$ seed matrix with a 
$3\times 3$ seed matrix. As can be seen in the experimental
section (Section~\ref{sec:exp}), the algorithm can
be very effective at generating intended degree distributions.
One subroutine in the algorithm involves sampling a $3\times 3$
seed from a $2\times 2$ seed. This approach can be generalized to
sampling an $m\times n$ seed matrix.

\subsection{Degree-2 and Degree-4 Optimization Problems}
\label{sec:degree-nm}

Let $S\in\reals^{a\times a}$ be a starting kronecker seed. Then to generate any graph of vertex size $a^t$ for any $t\in\naturals$,
we generate $S^{\otimes t}$ and sample an adjacency matrix from the
resulting matrix. As $t\rightarrow\infty$, $S^{\otimes t}$ yields a
probability distribution for a graph of any size. Using this fact and
a \textit{careful} analysis, we are able to sample seeds of relative
prime dimension from another Kronecker seed.

\begin{definition}[Kronecker Graph Distribution Function (KGD)]
Fix any starting seed $S\in\reals^{s\times s}$.
For all integer
$t\in[1,\infty)$, 
define $p_t = S^{\otimes t}\in\reals^{s^t\times s^t}$.
Define the normalization factor
$P_\infty = \int_{0}^\infty \int_{0}^\infty p_t(t_1, t_2) dt_1 dt_2$.

Let its KGD
$f:[0, 1]^2\times[0, 1]^2\rightarrow\reals$ be defined as:
\begin{enumerate}
\item For any
$(x_1, x_2), (y_1, y_2)\in[0, 1]^2$, 
$f((x_1, x_2), (y_1, y_2)) = f((0, x_2), (0, y_2)) - f((0, x_1), (0, y_1))$.
\item 
\begin{align}
&f((0, x), (0, y)) =\\
&\inf\{t_1^*\cdot t_2^*\,:\, (t_1^*, t_2^*)\in\reals^2, \frac{1}{P_\infty} \int_{0}^\infty\int_{0}^{t_1^*}p_t(t_1, t_2) dt_1 dt_2 \geq x, 
\frac{1}{P_\infty} \int_{0}^\infty\int_{0}^{t_2^*} p_t(t_1, t_2) dt_1 dt_2 \geq y\}.
\end{align}

\end{enumerate}

\label{def:kgd}
\end{definition}

Intuitively,
Definition~\ref{def:kgd} acts as an ``inverse CDF'' of the distribution
over $\lim_{t\rightarrow\infty} S^t$ for any starting seed $S$.
Also, $P_\infty$ is the area of $p_t$ and is used as a normalization
factor so that $f$ ranges from
$f((0, 0), (0, 0))$ to $f((0, 1), (0, 1))$.
Algorithm~\ref{alg:degree-nm} represents a general procedure for sampling 
an $m\times n$ seed from the uniform distribution over a starting
seed $T$. We then specialize this procedure to sample a
$3\times 1$ seed from a $2\times 1$ seed and sample a
$3\times 3$ seed from a $2\times 2$ seed.

\begin{algorithm}
  \KwData{$m, \ell, k, S$}
  Let $f:[0, 1]^2\times[0,1]^2\rightarrow\reals$ be the
  Kronecker Graph Distribution Function (KGD) of the seed $S$.

  Let $T$ be an $m\times n$ seed matrix\;

  \For{$i=1 \to m$} {
    \For{$j=1 \to n$} {
      $T[i, j] = f((\frac{i-1}{m}, \frac{i}{m}), (\frac{j-1}{n}, \frac{j}{n}))$
    }
  }
  \Return $T$\;
  \caption{Sampling $m\times n$ Seed}
  \label{alg:degree-nm}
\end{algorithm}

We also provide a closed-form solution that describes the desired $3\times 3$ matrix sampled from the uniform distribution over an \emph{arbitrary} $2 \times 2$ seed and use this closed-form in
Algorithm~\ref{alg:rpskg}. In that way we can create graphs with properties determined solely by the $2\times 2$ seed, as opposed to the NSKG model where the properties of the graph heavily vary depending on the noise parameters.

As a gentle start we will show how to compute the closed form solution in the 1-dimensional case and we will proceed to generalize it to the much more complicated 2-dimensional case.

\subsubsection{1-Dimensional Seed}

\begin{lemma}
There exists a polynomial-time
degree-2 optimization formulation to sample an SKG
$3\times 1$ seed from the uniform distribution over a
$2\times 1$ seed.
\label{lem:sample3x1}
\end{lemma}

\begin{proof}
Assume that we are given a 1-dimensional seed vector $T = [a\, b]$. What we would like to do is to create another vector with three entries corresponding to the integral of the uniform distribution over the intervals $[0, \frac{1}{3}], \, [\frac{1}{3}, \frac{2}{3}]$ and $[\frac{2}{3}, 1]$.

The uniform distribution over the seed can be defined as:
$p[0\ldots 1]$ is the uniform distribution over $\lim_{t \to \infty} T^{\otimes t}$, where $T^{\otimes t}$ is the Kronecker series of $T$ of size $t$. $T^{\otimes t}$ is a vector of size $2^t$. We can assume that the first coordinate of this vector is the value $p(0)$ and the last is the value $p(1)$. 
Let us denote with $f(x, y)$ the probability of $[x,y]$ according to $p$, i.e., $f(x,y) = \int_x^yp(t)dt$.

To compute the mass lying underneath the curve from 0 to 1/3, we first note that $\frac{1}{3} = \sum_{i=1}^\infty \frac{1}{4^i}$. Thus, computing $f(0, \frac{1}{3})$ is the same as computing the infinite series

$$f(0, \frac{1}{4}) + f(\frac{1}{4}, \frac{5}{16}) + f(\frac{5}{16}, \frac{21}{64}) + \cdots = \sum_{i=0}^\infty
f\left(\sum_{j=1}^i\frac{1}{4^j},
\sum_{j=1}^{i+1}\frac{1}{4^j}\right)
$$

Now note that these terms are very easily computed. The key insight is that in order to compute $f(0, \frac{1}{4})$ we just need to take the integral of the first quarter of the infinite tensor series. By taking the Kronecker product $T \otimes T = [a^2\, ab\, ba\, b^2]$ we basically learn the probability mass under the curve for any of the intervals $[0,\frac{1}{4}],\, [\frac{1}{4}, \frac{1}{2}], \, [\frac{1}{2}, \frac{3}{4}]$ and $[\frac{3}{4}, 1]$.

Thus $f(0, \frac{1}{4}) = a^2$. To compute $f(\frac{1}{4}, \frac{5}{16})$ we advance into the second quarter and then we take the first quarter again (i.e., the first quarter of the second quarter = fifth sixteenth of the matrix, which again is trivial to compute because it is an element of the matrix $T^{\otimes 4}$). Thus,  $f(\frac{1}{4}, \frac{5}{16}) = a^2 \cdot ab = a^3b$. Now we will move to the first quarter of the second quarter of the second quarter and repeat the same procedure. Eventually the infinite series that we get is:

$$a^2 + a^2\cdot ab + a^2 \cdot (ab)^2 + \cdots = \frac{a^2}{1 - ab},$$
where $ab < 1$ (without loss of generality since we can always re-scale the elements of the starting
kronecker seed). In the geometric series, $a^2$ is the starting term and the ratio of the series
is $ab$.

So now we have computed the first out of three coordinates of the vector. Using the exact same procedure we can compute the last one as well, by just reversing $a$ and $b$ and thus computing the mass of the distribution in $[2/3, 1]$, simply as $f(\frac{2}{3}, 1) = \frac{b^2}{1 - ab}$. Since $f(0,1)$ corresponds to the integral of a probability distribution we know that $f(0,1) = f(0, \frac{1}{3}) + f(\frac{1}{3}, \frac{2}{3}) + f(\frac{2}{3}, 1) = 1$ and thus we can also compute the value of $f(\frac{1}{3}, \frac{2}{3})$.

Another simple way to compute $f(\frac{1}{3}, \frac{2}{3})$ that will be useful in the 2-dimensional case as well is the following. Note that $\frac{2}{3} = 2\cdot \frac{1}{3} = 2\sum_{j=1}^\infty \frac{1}{4^j} = \sum_{j=0}^\infty\frac{1}{2^{2j+1}}$. Using the same technique as in the case of $f(0, \frac{1}{3})$, we can compute $f(0, \frac{2}{3}) = f(0, \frac{1}{2}) + f(\frac{1}{2}, \frac{5}{8}) + \cdots = \frac{a}{1- ab}$. So, $f(\frac{1}{3}, \frac{2}{3}) = f(0, \frac{2}{3}) - f(\frac{1}{3}, \frac{2}{3}) = \frac{a - a^2}{1 - ab}$.

Thus, we proved a closed form solution for sampling a 3-dimensional vector from the uniform distribution over a 2-dimensional vector. The case where our seed is 2-dimensional (i.e., a matrix and not a vector anymore) is 
more complicated and is analyzed in the following subsection.
Furthermore, starting from the $2\times 1$ seed
$$T = [a\quad b],$$
we obtain the $3\times 1$ seed
$$U = \left[\frac{a^2}{1-ab}\quad\frac{a-a^2}{1-ab}\quad\frac{b^2}{1-ab}\right].$$
These are clearly all degree-2 in the original variables.

\end{proof}

\subsubsection{2-Dimensional Seed}

\begin{lemma}
There exists a polynomial-time
degree-4 optimization formulation to sample an SKG
$3\times 3$ seed from the uniform distribution over a
$2\times 2$ seed.
\label{lem:sample3x3}
\end{lemma}

\begin{proof}

Now, assume that we are given a $2\times 2$ seed matrix
$T = \left[\begin{smallmatrix}
    a&b \\ c&d
    \end{smallmatrix} \right]$
 and we would like to create the $3\times 3$ matrix with values corresponding to the integrals of the uniform distribution associated with the seed, in the intervals $[0, \frac{1}{3}]\times [0, \frac{1}{3}], \, [0, \frac{1}{3}]\times [\frac{1}{3}, \frac{2}{3}]$, etc. The way the uniform distribution is defined, is completely analogous to the 1-D case. We generalize the semantics of $f$ in that case to as follows: $f((x_1, x_2), (y_1, y_2)) = \int_{y_1}^{y_2}\int_{x_1}^{x_2} p(t_1,t_2)dt_1dt_2$, where $p$ is the 2-dimensional uniform distribution associated with $T$.

We will use the same infinite series approach as before. This time however, in order to compute $f(0, \frac{1}{3})$, we will work on both dimensions simultaneously, expanding the $[0, \frac{1}{4}]\times [0, \frac{1}{4}]$ cell of the $4 \times 4$ matrix $T^{\otimes 2}$ until it reaches (asymptotically) the desired $[0, \frac{1}{3}]\times [0, \frac{1}{3}]$ tile of the $3 \times 3$ matrix $T$. This is more complicated since the cell needs to expand towards the right, towards down and towards the diagonal direction as well. That is: \begingroup\makeatletter\def\f@size{9}\check@mathfonts$$f\big((0, \frac{1}{3}), (0, \frac{1}{3})\big) = f\big((0, \frac{1}{4}), (0, \frac{1}{4})\big) + f\big((0, \frac{1}{4}), (\frac{1}{4}, \frac{5}{16})\big) + f\big((\frac{1}{4}, \frac{5}{16}), (0, \frac{1}{4})\big) + f\big((\frac{1}{4}, \frac{5}{16}), (\frac{1}{4}, \frac{5}{16})\big) + \cdots$$\endgroup

\noindent To make things more concrete, consider the $T^{\otimes 2}$ matrix (we denote with $\cdot$ the elements that are not useful in this computation):

\[T^{\otimes 2} =
 \begin{bmatrix}
    a^2 & ab & ab & b^2\\
    ac & ad & bc & bd\\
    ac & bc & \cdot & \cdot\\
    c^2 & cd & \cdot & \cdot\\
  \end{bmatrix}
\]

\medskip

To compute the series towards the right, i.e., $f\big((0, \frac{1}{4}), (\frac{1}{4}, \frac{5}{16})\big) + f\big((0, \frac{1}{4}), (\frac{5}{16}, \frac{21}{64})\big) + \cdots$, we start from the probability of being at the first column of the matrix and then recursively move to the second column. The probability of being at the first column is $a^2 + ac + ac + c^2 = (a+c)^2$, while the ratio of the series (i.e., the probability of being at the second column) is $ab + ad + cb + cd = (a+c)(b+d)$. Thus, analogously to the 1-D case, the infinite series towards the right is $\frac{(a+c)^2}{1 - (a+c)(b+d)}$.

Completely symmetrically, we can compute the infinite series towards down, by using the first row and recursively moving to the second row (instead of first and second column), i.e., $\frac{(a+b)^2}{1 - (a+b)(c+d)}$. 
Finally, the direction of expansion of the cell is towards the diagonal. The initial value of this series is

$$a^2 + ab \cdot \frac{(a+c)^2}{1 - (a+c)(b+d)} + ac \cdot \frac{(a+b)^2}{1- (a+b)(c+d)}$$

and the ratio of the series is just $ad$ (because in every step of the iteration we just move towards the $(2,2)$ cell of the matrix, which has probability $ad$). So eventually, putting everything together, we have that: 
$$f\big((0, \frac{1}{3}), (0, \frac{1}{3})\big) = \frac{a^2 + ab \cdot \frac{(a+c)^2}{1 - (a+c)(b+d)} + ac \cdot \frac{(a+b)^2}{1- (a+b)(c+d)}}{1- ad}$$

To make things further clear from now on let M be the $3\times 3$ matrix:

\[M =
 \begin{bmatrix}
    A & B & C\\
    D & E & F\\
    G & H & I\\
  \end{bmatrix}
\]

Note that similarly to the 1-D case we can get the values of all the corners, i.e., C, G and I, by transposing $T$ and using the formula for
$f\big((0, \frac{1}{3}), (0, \frac{1}{3})\big)$.

Next, we can compute the sum of the
probabilities in the first row

$$
A + B + C = f\big((0, \frac{1}{3}), (0, 1)\big) =
\frac{(a+b)^2}{1-(a+b)(c+d)}$$
and obtain $B$ as a result.

Note that the sum of the first row probabilities
is exactly equivalent to finding
$f(0, \frac{1}{3})$ in the 1-D case but the
initial term in
$f(0, \frac{1}{4})$ in this case is
$(a+b)^2$ (sum of first row in 2-D $T^{\otimes 2}$)
instead of $a^2$ (first entry in 1-D $T^{\otimes 2}$).

In a similar manner, we can compute

$$
G + H + I = f\big((\frac{2}{3}, 1), (0, 1)\big) =
\frac{(c+d)^2}{1-(a+b)(c+d)}
$$

$$
A + D + G = f\big((0, 1), (0, \frac{1}{3})\big) =
\frac{(a+c)^2}{1-(a+c)(b+d)}
$$

$$
C + F + I = f\big((0, 1), (\frac{2}{3}, 1)\big) =
\frac{(b+d)^2}{1-(a+c)(b+d)}
$$
and use these equations to compute
$H$, $D$, and $F$.

Finally, we can compute
the middle element $E = 1-(A + B + C + D + F + G + H + I)$ 
(up to normalizations, the entries in the sampled seed should sum to 1
but not necessarily in the original seed).
Thus, we can fully compute the $3\times 3$ seed
matrix $M$ based on the $2\times 2$
seed $T$.
Furthermore, it is degree-4 in the original variables. The
exact expression is used in Algorithm~\ref{alg:rpskg} as the
subroutine in Algorithm~\ref{alg:Sample3x3}.

\end{proof}

\subsection{The Algorithm}

Algorithm~\ref{alg:rpskg} details the procedure
used to generate a graph that gets rid of oscillations
present in the SKG model.
Parameters to this algorithm include:
$m, \ell, k, T$ that correspond to
the total number of edges generated, the number of
$2\times 2$ matrices used in the
Kronecker product,
the number of $3\times 3$ matrices used, and the
$2\times 2$ initiator seed matrix respectively.
Algorithm~\ref{alg:rpskg} uses Algorithm~\ref{alg:Sample3x3} as a sub-routine to sample from the
uniform distribution over the initiator matrix $T$ as $\ell\rightarrow\infty$.
In this section, we analyze 
both Algorithms~\ref{alg:rpskg}
and~\ref{alg:Sample3x3} for (computational) efficiency, 
ease of parallelization, and statistical impossibility when only
using $3\times 3$ seeds.

After creating the $3\times 3$ seed $M$
based on the uniform distribution over
the $2\times 2$ seed matrix,
Algorithm~\ref{alg:rpskg} generates each edge as
follows:
Choose $k$ out of $k+\ell$ random positions for
where we will insert the matrix $M$ in the
$k+\ell$ Kronecker sequence of matrices
for which we will obtain the Kronecker product.
Note that the
positions of insertions matter
because of the non-commutativity of the
Kronecker product for matrices of unequal
dimension. Then,
we simulate the Kronecker product of the $k$ copies
of the $M$ matrix and the $\ell$ copies of
the $T$ matrix. After simulation, we obtain an
edge $(u, v)$ which can then be inserted into the
generated edge set $E$.

On the other hand, Algorithm
~\ref{alg:Sample3x3} uses a closed-form
solution to obtain a $3\times 3$ seed
(with values $A, B, \ldots, H, I$)
from the
$2\times 2$ initiator seed matrix $T$. Refer to
Section~\ref{sec:degree-nm} for an explanation of
the derivation of this solution.

\begin{algorithm}
  \KwData{$m, \ell, k, T$}
  Initialize edge set $E = \{\}$\;
  $M \leftarrow $ Sample3x3($T$)\;
  \For{$j = 1 \to  k + \ell$} {
    positions[$j$] $\leftarrow 0$\;
  }
  \For{$i=1 \to m$} {
    // sampling without replacement\\
    Choose $k$
    random positions out of $k + \ell$ positions and set positions[$j$] = 1\;
    // zero-indexing used for vertex index generation\\
    $(u, v) \leftarrow (0, 0)$ \;

    \For{$j=1 \to k + \ell$}{
        \eIf{positions[$j$] = 1} {
            Pick random index $(\mu, \nu)$ of the $3\times 3$ matrix $M$ with probability proportional to its value\;$(u,v) \leftarrow (3u + \mu, 3v+ \nu)$\;
        } {
            Pick random index $(\mu, \nu)$ of the $2\times 2$ matrix $T$ with probability proportional to its value\;$(u,v) \leftarrow (2u + \mu, 2v + \nu)$\;
        }
    }

    $E \leftarrow E \cup \{(u, v)\}$\;
  }
  \Return $E$\;
  \caption{\rpskg Algorithm}
  \label{alg:rpskg}
\end{algorithm}

\begin{algorithm}
  \KwData{$T$}

  $\begin{bmatrix}
    a & b\\
    c & d\\
  \end{bmatrix} \leftarrow T$\;
  $A \leftarrow \frac{a^2 + ab \cdot \frac{(a+c)^2}{1 - (a+c)(b+d)} + ac \cdot \frac{(a+b)^2}{1- (a+b)(c+d)}}{1-ad}$\;
  $C \leftarrow \frac{b^2 + ab \cdot \frac{(b+d)^2}{1 - (a+c)(b+d)} + bd \cdot \frac{(a+b)^2}{1- (a+b)(c+d)}}{1-bc}$\;
  $G \leftarrow \frac{c^2 + cd \cdot \frac{(a+c)^2}{1 - (a+c)(b+d)} + ac \cdot \frac{(c+d)^2}{1- (a+b)(c+d)}}{1-bc}$\;
  $I \leftarrow \frac{d^2 + cd \cdot \frac{(b+d)^2}{1 - (a+c)(b+d)} + bd \cdot \frac{(c+d)^2}{1- (a+b)(c+d)}}{1-ad}$\;

  $B \leftarrow \frac{(a+b)^2}{1-(a+b)(c+d)} - (A + C)$\;
  $H \leftarrow \frac{(c+d)^2}{1-(a+b)(c+d)} - (G + I)$\;
  $D \leftarrow \frac{(a+c)^2}{1-(a+c)(b+d)} - (A + G)$\;
  $F \leftarrow \frac{(b+d)^2}{1-(a+c)(b+d)} - (C + I)$\;
  $E \leftarrow 1 - (A + B + C + D + F + G + H + I)$\;

  \Return $\begin{bmatrix}
    A & B & C\\
    D & E & F\\
    G & H & I\\
  \end{bmatrix}$\;
  \caption{Sample3x3: Used to generate $3\times 3$ matrix from the uniform distribution over $S$.}
  \label{alg:Sample3x3}
\end{algorithm}

\subsubsection{Runtime \& Parallelization}

Algorithm~\ref{alg:rpskg} is written in a serial
manner for clarity sake. But in practice, we can
parallelize the graph generation process of
\rpskg by using as many CPU/GPU cores as are available to
generate each edge.

We know that $m > \ell+k$ since
$m \geq n = 3^k\cdot2^\ell$.
When run serially, the runtime is $O(m(\ell+k))$. But
with parallelization, if we have $c$ CPU/GPU cores, then the
runtime becomes $O(\frac{m}{c}(\ell+k))$.

\subsection{Impossibility Results for Using Only $3\times 3$ Seeds}
\label{sec:threebythree}

In~\citep{Seshadhri:2013}, it is shown that the vanilla SKG model
leads to oscillations in the degree distribution, in certain
parameter regimes. The authors consider the case where a 
$2\times 2$ seed is used to generate the graph. We consider the
$3\times 3$ setting here and prove a similar result: the resulting
graph will also have oscillations. (This is not surprising since
this theoretical result is supported by our experimental results too.)

First, we recall the definition of slices in~\citep{Seshadhri:2013}:

\begin{definition}[Slices $S_r$ and Slice Probabilities $p_r$]
In the SKG model, $S_r$ denotes a \textit{slice} $r$
(with $r\in[-\ell/2, \ell/2]$)
and consists of all vertices whose binary
representations have exactly $(\ell/2 + r)$ zeros.

Furthermore, let $p_r$ be the probability that a single edge insertion
in the SKG generated graph produces an out-edge at node $v\in S_r$.
\end{definition}

Essentially, for any seed $T\in[0, 1]^2$,
every vertex in the graph generated by $T^{\otimes \ell}$ has a
corresponding
$\ell$-bit binary representation. This representation has a one-to-one
mapping to an element of the boolean hypercube $\{0, 1\}^\ell$.
We assume $\ell$ to be even for the analysis.
All vertices in the same slice have the same probability of
having an out-edge as proven by Seshadhri et al.~\cite{Seshadhri:2013}.

For purposes of discussion this section denote the $2\times 2$ seed
matrix entries as
$$
T = \begin{bmatrix}
    t_1 & t_2\\
    t_2 & t_3\\
  \end{bmatrix}.
$$

\begin{lemma}[Claim 3.3 in~\citep{Seshadhri:2013}]
For vertex $v\in S_r$ and $\ell\in\naturals$,
$$
p_r = \frac{(1-4\sigma^2)^{\ell/2}\tau^r}{n},
$$
where $\tau = \frac{1+2\sigma}{1-2\sigma}$,
$\sigma = t_1 + t_2 - 1/2$.
\end{lemma}

\begin{lemma}[Lemma 3.5 in~\citep{Seshadhri:2013}]
Let $v$ be a vertex in slice $r$. Assume that $p_r\leq 1/\sqrt{m}$ and
$d = o(\sqrt{n})$. Then for the original SKG model,
$$
\pr[\deg(v) = d] = (1+o(1))\frac{\lambda^d}{d!}\frac{(\tau^r)^d}{\exp(\lambda\tau^r)},
$$
where $\lambda = \Delta(1 - 4\sigma^2)^{\ell/2}$.
\label{lem:oscillations2times2}
\end{lemma}

Then as explained in~\citep{Seshadhri:2013},
Lemma~\ref{lem:oscillations2times2} implies oscillations in the
degree distribution since the lemma implies that the
probability that a vertex in slice $r$ has outdegree $d$ is
$$
\pr[\deg(v) = d] = (1+o(1))\frac{\exp(dr-e^r)}{d!},
$$
for $\lambda = 1$ and $\tau = e$.
Then since $r$ is discrete (not \textit{continuous!}) then for
$r = \lfloor \ln d\rceil$, $|r - \ln d|$ would fluctuate between
0 and 1/2, leading to an exponential tail when it is close to 1/2 and
no exponential tail when it is close to 0. This directly leads to
oscillations in the degree distribution.

We now derive a similar observation for the $3\times 3$ seed case.
First, we define some terms and parameters:

\begin{definition}
A slice $S_{\alpha,\beta}$ is the set of vertices
with $\alpha$ 0s, $\beta$ 1s, and $\ell-\alpha-\beta$
2s in the ternary representation of each vertex.

Furthermore, let $p_{\alpha, \beta}$ be the probability that a single edge insertion
in the SKG generated graph produces an out-edge at node $v\in S_{\alpha, \beta}$.
And let $p_\alpha = \sum_\beta p_{\alpha, \beta}$.
\end{definition}

First, let us define some parameters:
\begin{itemize}
\item $\sigma = t_1 + t_2 - 1/2$.
\item
$\tau = \frac{1/2+\sigma}{1/2 - \sigma}$,
\item
$\Delta = \frac{m}{n}$,
\item
$\Lambda = \Delta\left(\frac{3(1-2\sigma)^2}{3+4\sigma^2}\right)^\ell$.
\end{itemize}

As before,
$o(1)$ represents a negligible quantity as
$\ell\rightarrow\infty$.

Here, we discuss theory when only
$3\times 3$ seeds are used and prove that
(unwanted) oscillations are present in the
degree distribution of the generated graph.

This case is similar to the case where only $2\times 2$
seeds are used since
\begin{itemize}
\item Oscillations are present in
the degree distribution; and
\item Vertices belong to a single \textit{slice} throughout
the graph generation process.
\end{itemize}

The analysis for the 
$2\times 2$ case~\cite{Seshadhri:2013} used a binary
representation to represent each vertex.
For the $3\times 3$ case, we use a ternary representation
instead.
Recall that we sample a $3\times 3$ seed matrix (we
termed $M$) from
the uniform distribution over the $2\times 2$ matrix
(we termed $T$).
Again, we assume that $T$ is symmetric to make the analysis
clearer. Then the probabilities of an edge falling in
the first, second, and third rows/columns of $M$ are:
$$
\frac{(1/2 + \sigma)^2}{3/4+\sigma^2},
$$
$$
\frac{(1/2 + \sigma)(1/2 - \sigma)}{3/4+\sigma^2},
$$
and
$$
\frac{(1/2 - \sigma)^2}{3/4+\sigma^2}.
$$
respectively. We derive these values as follows.
Recall that, in Algorithm~\ref{alg:Sample3x3}, we sample the $3\times 3$ seed $M$ from a 
$2\times 2$ seed $T$:
$$
M = \begin{bmatrix}
    m_1 & m_2 & m_3\\
    m_2 & m_4 & m_5\\
    m_3 & m_5 & m_6\\
  \end{bmatrix}.
$$
where the probability of an edge becoming an out-edge
of a vertex in the first third is
\begin{align}
m_1 + m_2 + m_3
&= \frac{(t_1+t_2)^2}{1-(t_1 + t_2)(t_2 + t_3)}\\
&= \frac{(1/2 + \sigma)^2}{1-(1/2 + \sigma)(1/2 - \sigma)}\\
&= \frac{(1 + 2\sigma)^2}{3 + 4\sigma^2}
\end{align}
Similarly, the probabilities for the second and
third parts are
\begin{align}
m_2 + m_4 + m_5
&= \frac{(1/2 + \sigma)(1/2 - \sigma)}{1-(1/2 + \sigma)(1/2 - \sigma)}\\
&= \frac{1 - 4\sigma^2}{3 + 4\sigma^2}
\end{align}
and
\begin{align}
m_3 + m_5 + m_6
&= \frac{(1/2 - \sigma)^2}{1-(1/2 + \sigma)(1/2 - \sigma)}\\
&= \frac{(1 - 2\sigma)^2}{3 + 4\sigma^2}
\end{align}
respectively.
Note that
these three probabilities correspond to the $0_3, 1_3$
and $2_3$
ternary digit representations of a vertex in a specific
level of the edge generation process.

\begin{claim}
For any $v\in S_{\alpha, \beta}$,
let $p_{\alpha, \beta}$ be the probability that a single
edge insertion in SKG (on the $3\times 3$ seed) produces an
out-edge at node $v$. Then
$$
p_{\alpha, \beta} = \frac{\tau^{2\alpha + \beta}\Lambda}{\Delta n}.
$$
\label{claim:palphabeta}
\end{claim}

\begin{proof}

For an edge to be produced at $v\in S_{\alpha, \beta}$, it must
be that at every level, the first row/column of $M$ is considered
$\alpha$ times, the second is considered $\beta$ times, and
the third is considered $\ell-(\alpha+\beta)$ times. Then,

\begin{align}
p_{\alpha, \beta}
&=\frac{\left(\frac{1}{2} + \sigma\right)^{2\alpha}\left(\frac{1}{2} + \sigma\right)^\beta\left(\frac{1}{2}-\sigma\right)^\beta\left(\frac{1}{2}-\sigma\right)^{2(\ell-\alpha-\beta)}}{\left(\frac{3}{4} + \sigma^2\right)^\ell}\\
&=\frac{\left(\frac{1}{2} + \sigma\right)^{2\alpha + \beta}}{\left(\frac{1}{2} - \sigma\right)^{2\alpha + \beta}}
\cdot \left(\frac{\left(\frac{1}{2} - \sigma\right)^2}{\frac{3}{4} + \sigma^2}\right)^\ell
\\
&= \frac{\tau^{2\alpha + \beta}}{n}\cdot
\left(\frac{3(1-2\sigma)^2}{3+4\sigma^2}\right)^\ell\\
&=\frac{\tau^{2\alpha + \beta}\Lambda}{\Delta n}
\end{align}
since $n=3^\ell$.

\end{proof}

\begin{claim}
\label{eq:deg3x3}
Consider any vertex $v\in S_{\alpha, \beta}$. Let $d = o(\sqrt{n})$,
$p_\alpha \leq 1/\sqrt{m}$.
Then,
$$
\Prob{\text{deg}(v) = d} =
(1\pm o(1))\frac{1}{d!}\frac{(\tau^{2\alpha + \beta}\Lambda)^d}{\exp(\tau^{2\alpha + \beta}\Lambda)}
$$
\end{claim}

\begin{proof}

Observe that the outdegree of any vertex $v$ follows a binomial distribution, i.e., the probability that $v$ has outdegree $d$ is $\binom{m}{d}p_{\alpha, \beta}^d(1-p_{\alpha,\beta})^{m-d}$ where
$m$ is the total number of edges.

Now we can approximate $\binom{m}{d}$ by $m^d/d!$, since $d = o(\sqrt{n})$ (using Stirling's approximation for the factorial). We also use the Taylor series
approximation that $(1-x)^{m-d} = (1\pm o(1)) e^{-x(m-d)}$
for $x \leq 1/\sqrt{m}$.
for 
Last, observe that $p_{\alpha,\beta} \leq
\sum_{\beta}p_{\alpha, \beta} = p_\alpha \leq 1/\sqrt{m}$.

Then by Claim~\ref{claim:palphabeta},
\begin{align}
\Prob{\text{deg}(v) = d}
&= \binom{m}{d}p_{\alpha, \beta}^d(1-p_{\alpha,\beta})^{m-d} \\
&= (1\pm o(1))\frac{m^d}{d!}\left(\frac{\tau^{2\alpha + \beta}\Lambda}{\Delta n}\right)^d\exp\left(-\frac{\tau^{2\alpha + \beta}\Lambda(m-d)}{\Delta n}\right)\\
&= (1\pm o(1))\frac{1}{d!}\left(\frac{\tau^{2\alpha + \beta}\Lambda m}{\Delta n}\right)^d\exp\left(-\frac{\tau^{2\alpha + \beta}\Lambda m}{\Delta n}\right)\exp(p_{\alpha, \beta}d)\\
&= (1\pm o(1))\frac{1}{d!}\frac{(\tau^{2\alpha + \beta}\Lambda)^d}{\exp(\tau^{2\alpha + \beta}\Lambda)}
\end{align}

where $\exp(dp_{\alpha, \beta}) = o(1)$ since
$dp_{\alpha, \beta} = o(1)$ because $d = o(\sqrt{n})$.

\end{proof}

\subsubsection*{Exponential Tails in the Degree Distribution}

Seshadhri et al.~\cite{Seshadhri:2013} showed that
the degree distribution obtained from a
$2\times 2$ seed matrix oscillates between a
lognormal and exponential tail distribution.
We can do the same but for $3\times 3$ seeds.

Specifically, let us set $\tau = e$ and
$\Lambda = 1$ in the equation in
Claim~\ref{eq:deg3x3}:

\begin{theorem}
Assume that vertex $v\in S_{\alpha, \beta}$. Then,
$$
\pr[\deg(v) = d] = (1\pm o(1))\frac{\exp(d(2\alpha + \beta) - e^{2\alpha + \beta})}{d!}
$$
\label{thm:oscillations}
\end{theorem}

We obtain that the probability
of a vertex in slice $\alpha, \beta$ having
outdegree $d$ is
$$
\Prob{\text{deg}(v) = d} = (1\pm o(1))\frac{\exp(d(2\alpha + \beta) - e^{2\alpha + \beta})}{d!}
$$
which is the same as Equation 1 in~\cite{Seshadhri:2013} except that we replace
parameter $r$ with $2\alpha + \beta$ instead.
Applying Taylor approximations to appropriate
ranges of $2\alpha + \beta$ (we can treat as
one parameter), it can be shown that a suitable
approximation of the probability of slice
$\alpha, \beta$ having degree $d$ is roughly
$\exp(-d(2\alpha + \beta - \ln d)^2)$. Therefore only
vertices in slice $\alpha, \beta$ such that
$2\alpha + \beta \approx \ln d$ have a good chance
of having degree $d$.

Theorem~\ref{thm:oscillations} implies that using only
$3\times 3$ seeds still results in oscillations (even with
$\rpskg$). This justifies the use of another approach:
we introduce the mixing of $2\times 2$ and $3\times 3$ seeds.
Our experimental validation, in the next section,
shows that mixing seeds of relative
prime dimensions gets rids of oscillations.

\section{Experimental Validation}
\label{sec:exp}

\begin{figure}
    \begin{subfigure}[b]{0.5\textwidth}
        \includegraphics[width=\textwidth]{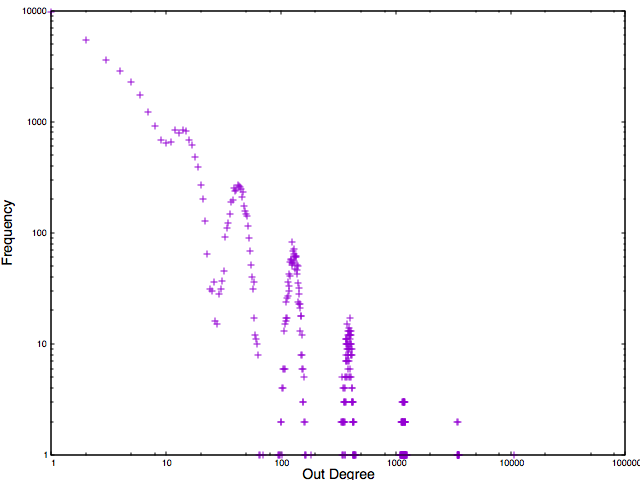}
        \caption{$\rpskg(2^{20}, 16, 0, T)$}
        \label{fig:s3=0}
    \end{subfigure}
    \begin{subfigure}[b]{0.5\textwidth}
        \includegraphics[width=\textwidth]{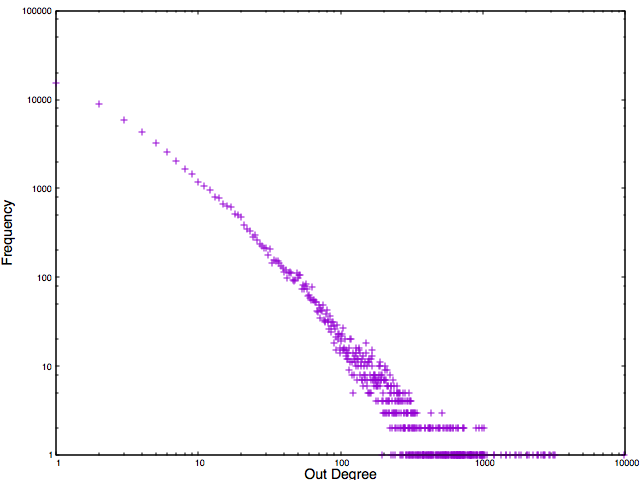}
        \caption{$\rpskg(2^{20}, 15, 1, T)$}
        \label{fig:s3=1}
    \end{subfigure}
    \label{fig:i0}
    \caption{Difference in degree distribution
    when only $2\times 2$ seeds are used
    versus when
    $3\times 3$ seeds are used with
     $2\times 2$ seeds.}
\end{figure}

\begin{figure}
    \begin{subfigure}[b]{0.5\textwidth}
        \includegraphics[width=\textwidth]{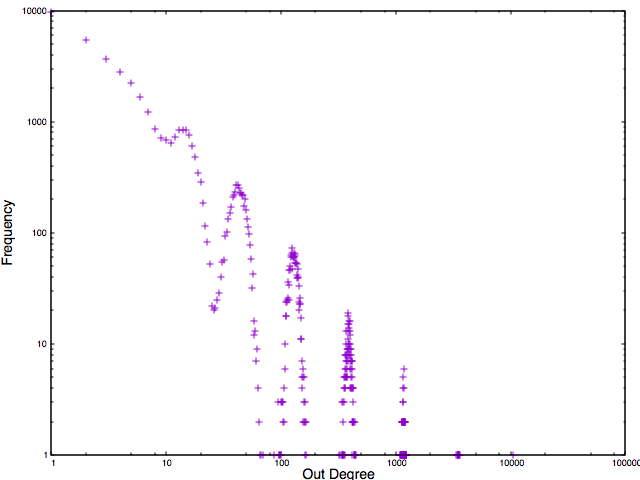}
        \caption{$\rpskg(2^{20}, 0, 16, T)$}
        \label{fig:s2=0}
    \end{subfigure}
    \begin{subfigure}[b]{0.5\textwidth}
        \includegraphics[width=\textwidth]{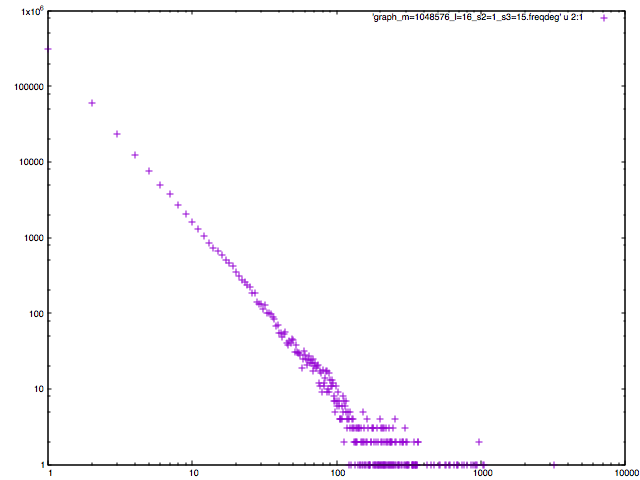}
        \caption{$\rpskg(2^{20}, 1, 15, T)$}
        \label{fig:s2=1}
    \end{subfigure}
    \label{fig:i1}
    \caption{Difference in degree distribution
    when only $3\times 3$ seeds are used
    versus when only one
    $2\times 2$ seed is used with
    multiple $3\times 3$ seeds}
\end{figure}

In this section, we discuss some of the experiments
we performed. We focus on comparing some of the
graphs generated by the SKG model to those
generated by the \rpskg model.
Note that
when $k=0$, Algorithm~\ref{alg:rpskg} is equivalent to
the SKG model --- when the $3\times 3$
seed is not used in any iteration of the edge
generation.

For our experiments, we applied
Algorithm~\ref{alg:Sample3x3} to the
Graph500 seed (see Equation~\ref{eq:Graph500})
to obtain the following
$3\times 3$ seed (approximated to 4 decimal places):

$$
M =
 \begin{bmatrix}
    0.4793 & 0.1598  & 0.0533\\
    0.1598 & 0.0533 & 0.0178\\
    0.0533 & 0.0178 & 0.0059\\
  \end{bmatrix}
$$

We note that users can
sample a $3\times 3$
seed from the uniform distribution over an
\textit{arbitrary} $2\times 2$ seed.
The procedure for producing the seed above is thus
flexible enough to accommodate additional graph
characteristics.

\begin{figure}
    \begin{subfigure}[b]{0.5\textwidth}
        \includegraphics[width=\textwidth]{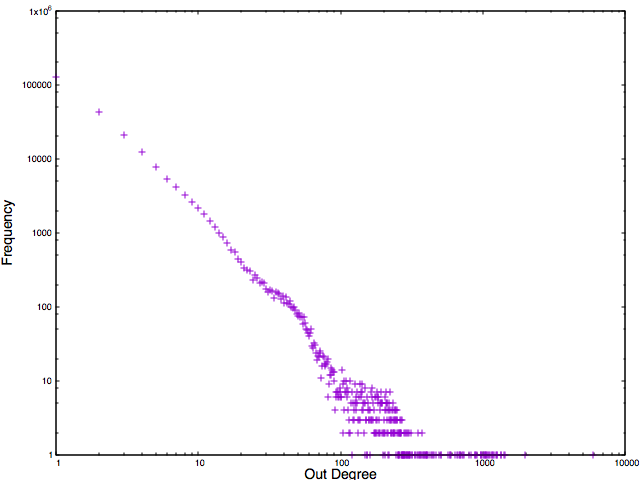}
        \caption{$\rpskg(2^{20}, 9, 7, T)$}
        \label{fig:s3=7}
    \end{subfigure}
    \begin{subfigure}[b]{0.5\textwidth}
        \includegraphics[width=\textwidth]{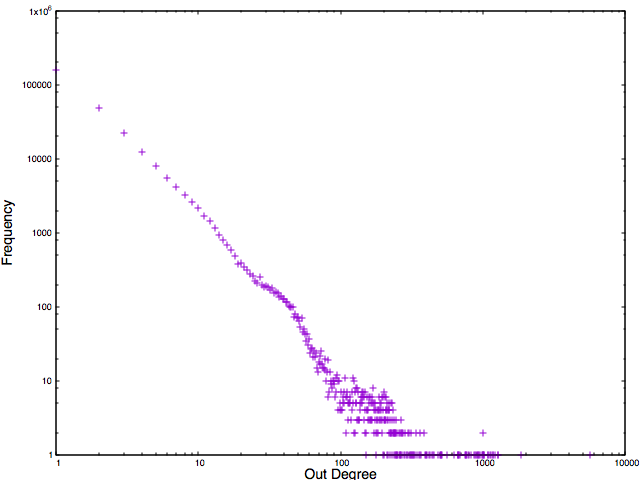}
        \caption{$\rpskg(2^{20}, 8, 8, T)$}
        \label{fig:s3=8}
    \end{subfigure}
    \begin{subfigure}[b]{0.5\textwidth}
        \includegraphics[width=\textwidth]{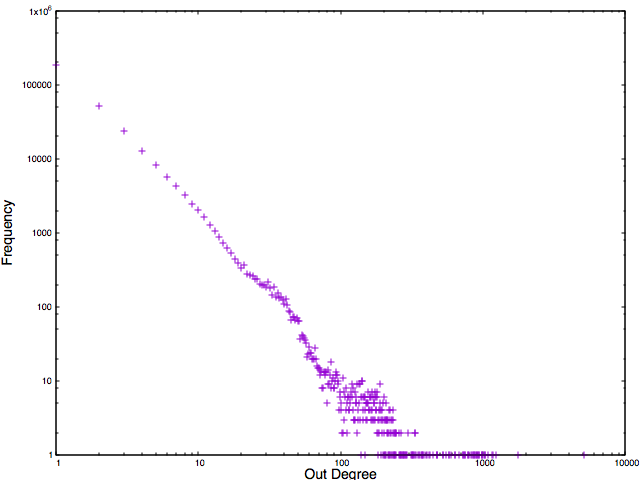}
        \caption{$\rpskg(2^{20}, 7, 9, T)$}
        \label{fig:s3=9}
    \end{subfigure}
    \label{fig:i1}
    \caption{Kinks in degree distribution observed
    when number of $3\times 3$ seeds used are
    close to number of $2\times 2$ seeds. $T$
    is the $2\times 2$ Graph500 initiator seed matrix given in Equation~\ref{eq:Graph500}.}
\end{figure}

\begin{figure}
    \begin{subfigure}[b]{0.5\textwidth}
        \includegraphics[width=\textwidth]{figs/graph_m=1048576_l=16_s2=1_s3=15}
        \caption{$\rpskg(2^{20}, 1, 15, T)$}
        \label{fig:s3=15}
    \end{subfigure}
    \begin{subfigure}[b]{0.5\textwidth}
        \includegraphics[width=\textwidth]{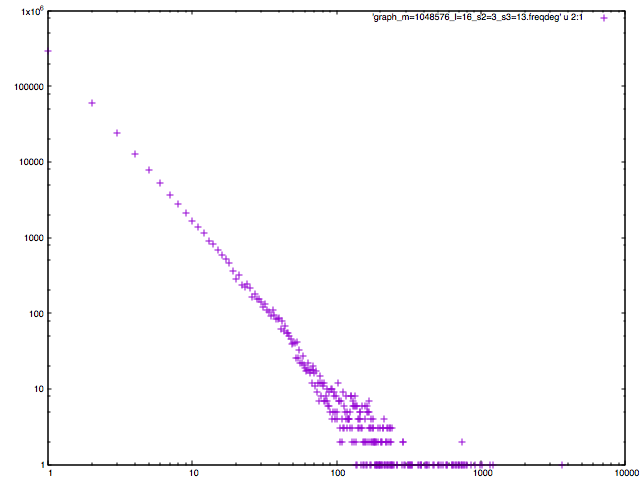}
        \caption{$\rpskg(2^{20}, 3, 13, T)$}
        \label{fig:s3=13}
    \end{subfigure}
    \caption{Smaller Kinks in degree distribution observed
    when the number of $3\times 3$ seeds used are not as
    close to number of $2\times 2$ seeds.}
\end{figure}

Figures~\ref{fig:s3=0} and~\ref{fig:s3=1}
illustrate the stark difference in the degree
distributions of graphs generated by \rpskg and SKG
models. Both plots are on a log-log scale and show
the frequency of nodes in the graph
for each possible (out-)degree --- in essence,
the degree distribution.

Figure~\ref{fig:s3=0} is a plot
of the degree distribution of a graph generated by
the SKG model. On the other hand,
Figure~\ref{fig:s3=1} is a plot of
the degree distribution when the \rpskg model is
used. Both models generate the same number of edges
($2^{20}$ edges) but differ in the number of vertices
generated. For the SKG model,
$2^{16}$ nodes are generated but for the \rpskg model,
$3\cdot 2^{15}$ nodes are generated.

As expected, Figure~\ref{fig:s3=0}
shows unwanted oscillations,
which according to Seshadhri et al.
~\cite{Seshadhri:2013} is most accurately
characterized as fluctuating between a log-normal
distribution and exponential tail.

On the other hand, using a single $3\times 3$
seed in \rpskg produces the graph shown in
Figure~\ref{fig:s3=1} with no salient oscillations.
This degree distribution plot seems to approximate
a log-normal distribution. In essence, 
it appears that
we are able
to achieve the same objective (in terms of the
generated degree distribution)
as NSKG without
generating or storing $\ell$ random numbers (needed
for the edge generation process of Noisy SKG).

Analogously,
Figures~\ref{fig:s2=0} and~\ref{fig:s2=1}
illustrate the stark difference in the degree
distributions of graphs generated by \rpskg and SKG
models but when $3\times 3$ seeds are the
primary seeds applied (as opposed to
the $2\times 2$ seeds).

Figures~\ref{fig:s2=0} and~\ref{fig:s3=0}
are almost indistinguishable, showing that
the use of only $2\times 2$ seeds or
only $3\times 3$ seeds will lead to
unwanted oscillations in the degree distribution.
Similar to Figure~\ref{fig:s2=1},
Figure~\ref{fig:s3=1} shows the degree
distribution when only one $2\times 2$ seed
is applied amidst $3\times 3$ seeds. The
degree distribution obtained in this figure
is approximately lognormal.

Furthermore, in our experiments on the graph with
$2^{20}$ edges we used
$k=1\ldots 16$ $3\times 3$ seeds
while keeping the total number of seeds used constant.
In essence, we ran
$\rpskg(2^{20}, 16-k, k, T)$ for $k=1\ldots 16$ where
$T$ is the $2\times 2$ Graph500 initiator seed
matrix (Equation~\ref{eq:Graph500}).
We noticed that the degree distributions of
generated graphs look more log-normal when only
a few number of either $2\times 2$ seeds or
$3\times 3$ seed are used --- when either
$k$ or $\ell$ is close to 1.

On the other hand, we noticed that when
$k$ is close to $\ell$, a
\emph{kink} starts to appear in the degree
distribution of the generated graph. Figures
~\ref{fig:s3=7},~\ref{fig:s3=8}, and
~\ref{fig:s3=9} show degree distributions for
graphs generated with $k=7, 8,$ and $9$
respectively. These three plots show a slight
kink in the degree distribution. Obviously,
these degree distributions
\textit{do not} exhibit
the unwanted oscillations generated
by SKG (see Figure~\ref{fig:s3=0}) but looks
less log-normal than the degree distribution in
Figure~\ref{fig:s3=1}.

\section{Conclusion}

Graph models can be used to model biological networks, social networks, communication
networks, relationships between products and advertisers, collaboration networks,
and even the world wide web.
However, major companies cannot
share all of their graph data because of copyright, legal, and privacy issues. This has led to the development of graph generation
algorithms, one of which is the family of stochastic kronecker
graph algorithms.

We have presented criteria for statistical and
computational identifiability of Stochastic Kronecker
Graphs. The criteria is used to conjecture that we can
classify SKG models into different classes and show
implications of such conjectures. Furthermore,
we also provide evidence for
the conjectures via simulations. 
The $\rpskg$ algorithm is also presented and can be used to
interpolate between the use of $2\times 2$ seeds and
$3\times 3$ seeds when generating graphs via stochastic
Kronecker multiplication. This algorithm can result in the
removal of oscillations in the resulting degree oscillations, thus
providing a superior case for its use. 

We believe that our work represents a starting point for further
understanding of graph generation algorithms and their properties
via a systematic notion of identifiability. Furthermore, here
are adjacent areas for exploration:

\paragraph{Separation Results}
We have been able to show separation of SKG models via the
statistical and computational identifiability notions. We studied
two prominent properties of graphs: the existence of isolated vertices
and the whether the graph is connected. However, we believe other
properties could also show (parameterized) separation across SKG models
and leave such exploration to future work.
Definition~\ref{def:kgd} defines a Kronecker graph distribution
function over any starting seed. We have defined efficient
procedures for sampling seeds from another seed of relative prime
dimension. It is clear that every graph has a Kronecker graph
distribution but it is not immediate, under what conditions,
every distribution has a corresponding realizable graph. We leave
this question to future exploration.

\paragraph{Privacy}
One of the main applications of the stochastic kronecker graph
model is to generate synthetic graph data. This is also a solution to
``machine unlearning'' now mandated in some some regulations.
\footnote{See Article 17 of the GDPR (EU General Data Protection Regulation) --- Right to erasure (``right to be forgotten''): \url{https://gdpr-info.eu/art-17-gdpr/}.}
The process and quality
of synthetic data is a major open problem in privacy-preserving
statistics and machine learning. We have identified
criteria (via the identifiability notions)
to measure the utility of generated graphs. However, what are the
necessary conditions (e.g., on graph size) to generate graphs that
are always within a given statistical distance?

\paragraph{Redistricting}
Recent work (e.g., ~\citep{McCartanImai23}) has shown random sampling
of graph partitions is a popular tool for evaluating
the effectiveness or bias of legislative redistricting plans.
Stochastic Kronecker Graphs, with accuracy measures and effective
degree distribution characteristics,
can be used to evaluate \textit{similar}
redistricting plans from the same starting seed. We leave the
applications of stochastic kronecker graphs for redistricting to
future exploration.

\clearpage

\bibliographystyle{alpha}
\bibliography{main}

\end{document}